\documentclass[11pt]{article}
\usepackage{authblk}
\usepackage{caption}
\usepackage[caption=false]{subfig}

\usepackage{amsmath,amsfonts,amsthm,amssymb,graphics,graphicx,mathtools,epstopdf,etoolbox,indentfirst,enumitem,mleftright,xcolor,booktabs,footnote,relsize}

\usepackage[scr=euler]{mathalpha} 

\makeatletter
\renewcommand{\paragraph}{%
  \@startsection{paragraph}{4}%
  {\z@}{1.25ex \@plus 1ex \@minus .2ex}{-1em}%
  {\normalfont\normalsize\bfseries}%
}
\makeatother

\usepackage[
backend=bibtex,
style=alphabetic,
giveninits=true, 
maxbibnames=99,
minalphanames=3,
date=year
]{biblatex}
\usepackage{hyperref} 
\usepackage{doi} 

\renewbibmacro{in:}{}

\AtEveryBibitem{\clearfield{issue}} 

\newbibmacro{string+doiurl}[1]{%
	\iffieldundef{doi}{%
		\iffieldundef{url}{
			#1%
		}{%
			\href{\thefield{url}}{#1}%
		}%
	}{%
		\href{http://doi.org/\thefield{doi}}{#1}%
	}%
}
\DeclareFieldFormat{title}{\usebibmacro{string+doiurl}{\mkbibemph{#1}}}
\DeclareFieldFormat[article,thesis,incollection,inproceedings,manual]{title}%
{\usebibmacro{string+doiurl}
	{#1} 
}
\usepackage{algorithm} 
\usepackage[noend]{algpseudocode}
\floatname{algorithm}{Protocol} 
\algrenewcommand\alglinenumber[1]{\normalsize #1.} 

\newcounter{algsubstate}




\definecolor{darkmagenta}{rgb}{0.85, 0, 0.45}
\hypersetup{
colorlinks = true,
citecolor= blue,
urlcolor= blue,
linkcolor = darkmagenta
}

\newcommand{\ket}[1]{\left| #1 \right>}
\newcommand{\bra}[1]{\left< #1 \right|}
\newcommand{\ketbra}[2]{\ket{#1} \!\! \bra{#2}}
\newcommand{\pure}[1]{\ketbra{#1}{#1}}
\newcommand{\inn}[2]{\langle#1|#2\rangle} 
\newcommand{\tr}[2][]{\operatorname{Tr}_{#1}\!\left[#2\right]} 
\newcommand{\binh}{h_2} 

\newcommand{\acos}{\cos^{-1}}

\newcommand{\constr}{\mu} 
\newcommand{\defvar}{\coloneqq} 
\newcommand{\dop}[1]{\operatorname{S}_{#1}} 
\newcommand{\eps}{\epsilon}

\newcommand{\Hmin}{H_\mathrm{min}}
\newcommand{\Hmax}{H_\mathrm{max}}
\newcommand{\id}{\mathbb{I}} 
\newcommand{\idmap}{\mathcal{I}} 

\newcommand{\norm}[1]{\left\lVert#1\right\rVert} 
\newcommand{\pd}{P} 

\newcommand{\pvm}{P} 
\newcommand{\q}{q} 
\newcommand{\str}[1]{\mathbf{#1}} 
\newcommand{\smf}[1]{\vartheta_{#1}} 
\newcommand{\suchthat}{\text{ s.t.}} 
\newcommand{\term}[1]{\emph{#1}}

\newcommand{\esecr}{\eps^\mathrm{sec}}

\newcommand{\es}{{\eps_s}}
\newcommand{\esp}{{\eps'_s}}

\newcommand{\esL}{\nu} 

\DeclareMathAlphabet{\mathpzc}{OT1}{pzc}{m}{it} 

\newcommand{\Renyi}{R\'{e}nyi}

\newcommand{\qA}{Q^A}
\newcommand{\qB}{Q^B}

\newcommand{\En}{\mathsf{E}}

\newcommand{\mA}{M^A}
\newcommand{\mB}{M^B}
\newcommand{\lchann}{\mathcal{L}}
\newcommand{\lk}{L}
\newcommand{\lkE}{\lk}
\newcommand{\gnd}{\phi}
\newcommand{\dleak}{\delta_\mathrm{leak}} 
\newcommand{\mchann}{\mathcal{M}} 

\newcommand{\ptest}{p^\mathrm{test}}
\newcommand{\pgen}{p^\mathrm{gen}}
\newcommand{\fcont}{f_\mathrm{cont}}

\newcommand{\dom}{D}
\newcommand{\dL}{d_{\lkE}}
\newcommand{\dualf}{g}
\newcommand{\lagG}{\kappa} 
\newcommand{\lagE}{\beta} 
\newcommand{\lagP}{\lambda} 
\newcommand{\avgE}{E_\mathrm{exp}}
\newcommand{\cutE}{E_\mathrm{cutoff}}

\newtheorem{remark}{Remark}

\newtheorem{lemma}{Lemma}

\theoremstyle{definition} 
\newtheorem{definition}{Definition}

\newfloat{process}{htbp}{loa}
\floatname{process}{Process}

\begin{document}

\title{\textbf{Robustness of implemented device-independent protocols against constrained leakage}}
\renewcommand\Affilfont{\itshape\small} 

\author[1]{Ernest Y.-Z.\ Tan}
\affil[1]{Institute for Quantum Computing and Department
of Physics and Astronomy, University of Waterloo, Waterloo, Ontario N2L 3G1, Canada.}

\date{}

\maketitle

\begin{abstract}
Device-independent (DI) protocols have experienced significant progress in recent years, with a series of demonstrations of DI randomness generation or expansion, as well as DI quantum key distribution. However, existing security proofs for those demonstrations rely on a typical assumption in DI cryptography, that the devices do not leak any unwanted information to each other or to an adversary. This assumption may be difficult to perfectly enforce in practice. While there exist other DI security proofs that account for a constrained amount of such leakage, the techniques used are somewhat unsuited for analyzing the recent DI protocol demonstrations. In this work, we address this issue by studying a constrained leakage model suited for this purpose, which should also be relevant for future similar experiments. Our proof structure is compatible with recent proof techniques for flexibly analyzing a wide range of DI protocol implementations. With our approach, we compute some estimates of the effects of leakage on the keyrates of those protocols, hence providing a clearer understanding of the amount of leakage that can be allowed while still obtaining positive keyrates.
\end{abstract}

\section{Introduction}

Device-independent (DI) cryptography is the concept of exploiting Bell inequality violations from quantum devices to achieve cryptographic tasks~\cite{BHK05,PAB+09,Sca12}. 
Informally, it relies on the observation that if two or more devices are not allowed to communicate between each other, then the only way for them to violate a Bell inequality is for them to be performing some ``genuinely quantum'' operations; therefore, certifying a Bell inequality violation ensures that some form of quantum behaviour is occurring in the devices, which could potentially be used for cryptographic purposes. The critical point about this reasoning is that it holds regardless of what states and measurements are being implemented in the devices, hence the term ``device-independent'' --- informally, the only assumption made on the devices in a DI security proof is that they do not communicate information outside of the protocol's specifications. This is in contrast to more ``standard'' quantum cryptography protocols, which are ``device-dependent'' in the sense that they rely on the honest parties' device measurements (and/or state preparation) being well-characterized. DI cryptography hence offers a path towards quantum cryptography that is robust to a wide class of device imperfections, since it remains secure even if the states and measurements deviate from the intended ones.

In the past years, a number of theoretical~\cite{DFR20,ZKB18,ZFK20} and experimental~\cite{HBD+15,SMC+15,GVW+15,RBG+17} advances have led to significant development in DI security proofs and protocol implementations. In particular, two families of DI protocols have undergone especially notable progress: \term{device-independent randomness expansion} (DIRE), in which the goal is to expand a short secret string into a longer one, and \term{device-independent quantum key distribution} (DIQKD), in which the goal is for two parties to establish a shared secret key. There have been several recent experimental demonstrations of DIRE~\cite{LLR+21,SZB+21,LZL+21} (and the related task of DI randomness generation~\cite{LZL+18,ZSB+20}), and progress has also been made on DIQKD demonstrations~\cite{NDN+22,ZLR+22,LZZ+22}, though with somewhat worse performance as compared to DIRE.

In light of the above, it is important to ensure that the security proofs for these DI protocols are based on as accurate a model of the physical devices as possible. One limitation of most existing DI security proofs is that they rely on the assumption that absolutely no unwanted information ``leaks'' from the devices: this assumption is used not only to enforce the condition that the devices do not communicate in order to achieve the Bell violation, but also to ensure that the raw data produced by the devices is not simply broadcasted to an adversary. While this is already a fairly weak assumption as compared to device-dependent quantum cryptography, it is perhaps somewhat unrealistic as an absolute condition --- one would expect that in a given physical implementation, some small amount of information could potentially leak out from the devices. 
It is hence important to ensure that DI protocols have some ``robustness'' against a limited amount of such leakage: while it is certainly intuitive that such robustness should hold, past observations such as the phenomenon of \term{information locking}~\cite{KRBM07,DHL+04,Win17} indicate that some care is needed here, as there can be situations where a small amount of additional information can have an unexpectedly large effect. 

There have been a few previous works studying the effects of leakage or communication between devices in DI protocols; however, the models they have used are not a very close fit for describing the recent experimental implementations of DIRE or DIQKD. For instance,~\cite{SPM13} studied a model for weak cross-talk between the devices, but it is not straightforward to extend that model to account for leakage to an adversary when the device behaviour is not independent and identically distributed (IID) across protocol rounds. A model that accounts for non-IID behaviour and leakage to an adversary was studied in~\cite{arx_JK21}, but the leakage constraint in that model is that only a bounded number of qubits\footnote{The number of leakage qubits can be linear with respect to the number of protocol rounds $n$, but must be smaller than the amount of smooth min-entropy that would be produced without leakage (see Sec.~\ref{sec:fullprot} for more details).} is leaked between the devices and/or adversary over the course of the protocol. While this is a clean model to study, in the context of experimental implementations it is not straightforward to rigorously formulate a nontrivial upper bound on the number of qubits leaked during a protocol --- for instance, if the leakage occurs via photonic systems with some classical Poissonian number distribution, then in theory the state has infinite-dimensional support. (Of course, qualitatively one would expect that if for instance the state has a large vacuum component, then it should not leak too much ``useful information'', but formalizing this idea is part of the goal in this work.) In a somewhat different direction,~\cite{MS14,CL19,LRR19} have studied methods to mitigate leakage arising from device-reuse attacks~\cite{BCK13} and malicious classical post-processing units, but this is mostly focused on leakage that occurs after the protocol has finished distributing and measuring the states, rather than leakage that occurs during that process.

We also note that outside of DI cryptography, there have been studies of modified prepare-and-measure scenarios that can potentially be viewed as describing some form of untrusted leakage between the devices~\cite{TZWP22,TPWP21,PPWT22}; however, the setups studied in those works are currently somewhat different from those in DI cryptography (though we note that the ``bounded-weight'' model we describe later in this work has similarities to the model used in~\cite{PPWT22}). Also, for device-dependent QKD in particular, the possibility of leakage due to detector backflash was considered in e.g.~\cite{KZMW01,PCS+18}. That model is quite similar to what we consider in this work, but we extend the analysis to the DI scenario and handle a number of complications that arise when non-IID behaviour is allowed (see Sec.~\ref{sec:leakmodel}). 

The contribution of this work is to study a model for constrained leakage from the devices in a DI protocol, suited for analyzing existing DIRE and DIQKD demonstrations (though the analysis may generalize to some other DI protocols). 
Qualitatively, the idea is that in each round we assign some registers that model the leakage processes, and impose the constraint that with high probability the leakage registers are in some fixed ``blank'' state.
While this model is fairly simple, it seems possible to estimate such leakage probabilities in some DI experimental setups~\cite{NDN+22}, and we note that (as observed in~\cite{SPM13}) allowing for leakage in any sense is already covering a strictly wider class of scenarios than is typical in DI security proofs, or for that matter most device-dependent security proofs.
Our approach for analyzing this leakage model takes place in roughly two parts. First, we analyze its effect on single rounds of the protocol, essentially by arguing that in this constrained leakage model, the state in each round is ``close'' to one where no leakage occurred. Next, we describe how this can then be converted into a security proof for the full protocol by using a series of entropic chain rules, taking into account non-IID effects (in particular, the possibility that leakage in later rounds could contain information about the secret data produced in earlier rounds).  
We also remark that although the focus of this work is on DI cryptography, the techniques described here should generalize to a variety of device-dependent protocols as well, such as QKD or QRNG.

This paper is structured as follows. In Sec.~\ref{sec:prelim}, we introduce notation and specify the leakage model, as well as describing the overall proof structure. In Sec.~\ref{sec:1rnd} we present the analysis of single rounds, and in Sec.~\ref{sec:fullprot} we describe how to obtain a security proof for the full protocol; in each section we explicitly compute some examples showing how much the keyrates are reduced by the leakage effects. Finally, in Sec.~\ref{sec:conclusion} we describe some potential future directions to explore.

\section{Preliminaries}
\label{sec:prelim}

\subsection{Notation and definitions}
\begin{table}[h!]
\caption{List of notation}
\def\arraystretch{1.5} 
\setlength\tabcolsep{.28cm}
\centering
\begin{tabular}{c l}
\toprule
\textit{Symbol} & \textit{Definition} \\
\toprule
$\log$ & Base-$2$ logarithm \\
\hline
$H(\cdot)$ & Base-$2$ von Neumann entropy \\
\hline
$\norm{\cdot}_p$ & Schatten $p$-norm \\
\hline
$\dop{=}(A)$ (resp.~$\dop{\leq}(A)$) & Set of normalized (resp.~subnormalized) states on register $A$ \\
\hline
$A_j^k$ & Registers $A_j A_{j+1} \dots A_{k-1} A_k$
\\
\toprule
\end{tabular}
\def\arraystretch{1}
\label{tab:notation}
\end{table}

For technical reasons (as some of the theorems we use may not have been proven yet for infinite-dimensional systems) we take all systems to be finite-dimensional, but we will not impose any bounds on the system dimensions unless otherwise specified, i.e.~they can have unboundedly large finite dimension.

\begin{definition}
For $\rho,\sigma \in \dop{\leq}(A)$, the \term{generalized fidelity} is
\begin{align}
F(\rho,\sigma) \defvar \norm{
\sqrt{\rho}\sqrt{\sigma}
}_1 + \sqrt{(1-\tr{\rho})(1-\tr{\sigma})},
\end{align}
and the \term{purified distance} is $\pd(\rho,\sigma)\defvar\sqrt{1-F(\rho,\sigma)^2}$.
\end{definition}
\noindent (Note that this means we are using the convention that for normalized pure states, we have $F(\pure{\psi},\pure{\phi}) = |\inn{\psi}{\phi}|$, not $|\inn{\psi}{\phi}|^2$.)

We now state the definitions of various entropies required in our analysis.
We follow the presentation in~\cite{DFR20}, which can be shown to be equivalent to the definitions in~\cite{Tom16}.

\begin{definition}\label{def:HminHmax}
For $\rho\in\dop{\leq}(AB)$, the \term{min- and max-entropies of $A$ conditioned on $B$} are
\begin{align}
\Hmin(A|B)_\rho &\defvar 
-\log 
\min_{\substack{\sigma \in \dop{\leq}(B) \suchthat\\ \ker(\rho_B)\subseteq\ker(\sigma_B)}} 
\norm{\rho_{AB}^\frac{1}{2}
(\id_A \otimes \sigma_{B})
^{-\frac{1}{2}}}_\infty^2,
\\ 
\Hmax(A|B)_\rho &\defvar \log 
\max_{\sigma \in \dop{\leq}(B)} 
\norm{\rho_{AB}^\frac{1}{2}
(\id_A \otimes \sigma_{B})
^\frac{1}{2}}_1^2, 
\end{align}
where in the first equation the 
$(\id_A \otimes \sigma_{B})
^{-\frac{1}{2}}$ term
should be understood in terms of the Moore-Penrose generalized inverse.
(Note that the optimum is indeed attained in both equations~\cite{Tom16}, and it can be attained by a normalized state 
so $\dop{\leq}(B)$ can be replaced by $\dop{=}(B)$ without loss of generality.)
\end{definition}

\begin{definition}\label{def:Hsmooth}
For $\rho\in\dop{\leq}(AB)$ and $\es\in\left[0,\sqrt{\tr{\rho_{AB}}}\right)$, the \term{$\es$-smooth min- and max-entropies of $A$ conditioned on $B$} are
\begin{align}
\Hmin^\es(A|B)_\rho \defvar
\max_
{\substack{\tilde{\rho} \in \dop{\leq}(AB) \suchthat\\ \pd(\tilde{\rho},\rho)\leq\es}}
\Hmin(A|B)_{\tilde{\rho}}, 
\qquad
\Hmax^\es(A|B)_\rho \defvar
\min_
{\substack{\tilde{\rho} \in \dop{\leq}(AB) \suchthat\\ \pd(\tilde{\rho},\rho)\leq\es}} 
\Hmax(A|B)_{\tilde{\rho}}.
\end{align}
\end{definition}

\begin{definition}\label{def:Renyi}
For $\rho\in\dop{=}(A)$ and $\alpha \in (0,1) \cup (1,\infty)$, the \term{$\alpha$-{\Renyi} entropy of $A$} is
\begin{align}
H_\alpha(A)_\rho \defvar 
\frac{1}{1-\alpha} \log \norm{\rho_A}^\alpha_\alpha.
\end{align}
\end{definition}
\noindent In the $\alpha\to 1$ limit, the {\Renyi} entropy reduces to the von Neumann entropy. (The {\Renyi} entropy can be extended to include conditioning systems in multiple different ways~\cite{Tom16}, but we will not require those in this work; all those definitions match the above expression when there is no conditioning system.)

\subsection{Leakage model}
\label{sec:leakmodel}

We shall suppose (as is the case in DIRE and DIQKD) that the protocol begins with Alice and Bob performing $n$ sequential rounds of supplying some classical inputs to their devices and receiving some classical outputs. To account for leakage from the devices over this part of the protocol, we shall consider the following model.
(The specific order of events described here may appear slightly restrictive, but the analysis remains essentially similar if we consider some more general versions; see Appendix~\ref{app:varorder} for further discussion.)
We use the following notation:
for the $j^\text{th}$ round, $\qA_j$ (resp.~$\qB_j$) denotes the quantum register that will be measured in Alice's (resp.~Bob's) device, $X_j$ (resp.~$Y_j$) denotes the input supplied to the device, $A_j$ (resp.~$B_j$) denotes the output obtained, and $\mA_j$ (resp.~$\mB_j$) denotes a memory register the device can retain from previous rounds.
We require the classical registers storing the input and output values to have some known finite dimension; the other registers can be of unknown dimension.
We also introduce several registers $\lk^{A\to B}_j, \lk^{A\to E}_j, \lk^{B\to A}_j, \lk^{B\to E}_j$ to track the ``leakage'' processes we are about to describe.
Finally, let $\En$ denote a register Eve holds to collect quantum side-information across all the rounds (she can update this register as each round occurs, as we shall describe below --- in principle we could use a different register to denote her updated side-information after each such process, following e.g.~\cite{arx_MFSR22}, but as we allow the dimension of $\En$ to be unbounded, there is no loss of generality by just discussing this one register).
We then model the physical process in each round as follows:

\begin{enumerate}
\item A state preparation process takes place, modelled as follows: Eve first performs some channel $\En \to \En \qA_j \qB_j$, and then some other ``update'' channel\footnote{Here we have allowed this channel to act across the memory registers of \emph{both} devices --- the proof approach we shall use is compatible with such a structure~\cite{DFR20,ARV19,arx_MFSR22}, so we include this possibility for generality, even if it may not correspond to some intuitive physical process.} $\mA_j \mB_j \qA_j \qB_j \to \qA_j \qB_j$ is performed using the memory registers retained from the previous round, with the registers $\qA_j$ and $\qB_j$ at the end of this process being held in Alice and Bob's devices respectively.

\item Alice prepares an input $X_j$ to supply to her device, which performs some ``leakage channel'' $\lchann_j^A : \qA_j X_j \to \qA_j \lk^{A\to B}_j \lk^{A\to E}_j X_j$ that does not disturb\footnote{This no-disturbance condition is just for ease of discussion in our subsequent analysis, so that we only need a single register $X_j$ to keep a ``persistent'' record of Alice's input throughout the protocol. 
In principle, we could have instead said more formally that Alice copies the classical register $X_j$ onto another register $\hat{X}_j$ that is supplied to the device, which then performs some channel $\qA_j \hat{X}_j \to \qA_j \lk^{A\to B}_j \lk^{A\to E}_j$ without involving the original $X_j$ register. However, for brevity we will usually just say that this overall process is a channel that ``does not disturb $X_j$''.
} the classical register $X_j$.
Analogously, Bob supplies an input $Y_j$ to his device, which performs some leakage channel $\lchann_j^B : \qB_j Y_j \to \qB_j \lk^{B\to A}_j \lk^{B\to E}_j Y_j$ that does not disturb $Y_j$. 
For brevity, we shall write the overall leakage channel as $\lchann_j \defvar \lchann_j^A \otimes \lchann_j^B$. The registers $\lk^{A\to E}_j  \lk^{B\to E}_j$ are now sent to Eve, while $\lk^{B\to A}_j$ is sent to Alice's device and $\lk^{A\to B}_j$ is sent to Bob's device.

\item Alice's device performs some uncharacterized measurement channel $\mchann^A: \qA_j \lk^{B\to A}_j X_j \to A_j \mA_{j+1} X_j $ that does not disturb $X_j$, where $A_j$ is a classical register storing the measurement outcome and $\mA_{j+1}$ is the memory register retained for the next round. Analogously, Bob's device receives $\lk^{A\to B}_j$ and performs some uncharacterized measurement channel $\mchann^B: \qB_j \lk^{A\to B}_j Y_j \to B_j \mB_{j+1} Y_j$ that does not disturb $Y_j$.
Alice and Bob then announce their inputs $X_j Y_j$, and Eve can use those values to update her register $\En$.\footnote{Technically, depending on the exact protocol specification, Alice and Bob might not announce their inputs after each round. However, this is a rather specialized discussion and we defer the details to Appendix~\ref{app:EAT}.}
\end{enumerate}
\noindent Note that in the above description, we have imposed a subtle restriction on Eve --- specifically, while we allow her to update the register $\En$ (and thus the state preparation for the next round) using the values $X_j Y_j$, we do not include the possibility of her using the leakage registers $\lk^{A\to E}_j  \lk^{B\to E}_j$ as well when doing so in each round. This ``restricted adaptiveness'' condition is currently required for our proof approach, as we discuss in more detail (along with some motivating circumstances) in Sec.~\ref{sec:fullprot} later. 

If the leakage registers are unconstrained,
any attempt at DI cryptography is futile --- Bell violations can be trivially faked by using the registers $\lk^{A\to B}_j \lk^{B\to A}_j$ to communicate either party's input to the other's device, and furthermore the registers $\lk^{A\to E}_j \lk^{B\to E}_j$ could just give copies of all the device outputs to Eve.
However, note that the standard assumption in DI cryptography is equivalent to stating that all these leakage registers are trivial (or at least independent of the inputs/outputs), which might be considered to be too extreme of an assumption. Hence in this work, we study a scenario where we impose a more relaxed constraint on these registers. 
Specifically, we consider the following constraint. (As noted in the introduction, similar ideas have been explored in the context of device-dependent QKD, e.g.~in~\cite{KZMW01,PCS+18}. However, here we consider the DI case; furthermore there are some technical issues to address for non-IID leakage --- see Remark~\ref{remark:nonIID} below.)

\paragraph{Bounded-weight leakage constraint:}
We suppose we have certified some value $\dleak>0$ such that all $\lchann_j$ have the following property:
there exists some state 
$\pure{\gnd}$
such that if the measurement described by projectors $\left(\pure{\gnd}^{\otimes 4}, 
\id
-\pure{\gnd}^{\otimes 4}\right)$ is performed on the $\lk^{A\to B}_j \lk^{A\to E}_j \lk^{B\to A}_j \lk^{B\to E}_j$ registers in any state produced by $\lchann_j$, the probability of getting the outcome $\pure{\gnd}^{\otimes 4}$ is always at least $1-\dleak$ (regardless of the input state to $\lchann_j$).\\[-3mm]

Qualitatively, an example of physical reasoning behind such a constraint could be if $\pure{\gnd}$ is a ground state of those registers, and we have certified that if we measure those registers in some basis that includes the ground state as a possible outcome, we will obtain the ground state with high probability. 
Note that the bounded-weight constraint also straightforwardly implies (see Appendix~\ref{app:varprob}) that if e.g.~we consider only the registers $\lk^{A\to B}_j \lk^{B\to A}_j$ and perform the projective measurement $\left(\pure{\gnd}^{\otimes 2}, \id-\pure{\gnd}^{\otimes 2}\right)$, then the probability of getting the outcome $\pure{\gnd}^{\otimes 2}$ is at least $1-\dleak$ as well (and analogous statements hold for any other subset of the registers $\lk^{A\to B}_j \lk^{A\to E}_j \lk^{B\to A}_j \lk^{B\to E}_j$). 
(A minor variant of the bounded-weight constraint would be to impose an analogous condition on each of the registers individually, but that yields basically the same results up to rescaling $\dleak$ by a factor of $4$; we discuss it further in Appendix~\ref{app:varprob}.) 

We remark that the bounded-weight leakage model is somewhat more general than a constraint of the form ``with probability at least $1-\dleak$, the channel $\lchann_j$ acts locally in the devices and independently sets 
all leakage registers
to $\pure{\gnd}$''. Such a constraint would be more correctly expressed as the following version, which we shall term as ``classical-probabilistic leakage'':

\paragraph{Classical-probabilistic leakage constraint:}
We suppose we have certified some value $\dleak>0$ such that the following holds:
there exist channels 
$\hat{\lchann}^A_j: \qA_j X_j \to \qA_j X_j$ and
$\hat{\lchann}^B_j: \qB_j Y_j \to \qB_j Y_j$ and
$\hat{\lchann}_j: \qA_j \qB_j X_j Y_j \to \qA_j \qB_j \lk^{A\to B}_j \lk^{A\to E}_j \lk^{B\to A}_j \lk^{B\to E}_j X_j Y_j$ (that all do not disturb the classical registers $X_j Y_j$), 
such that all $\lchann_j$ have the form
\begin{align}\label{eq:classprob}
\lchann_j\!\left[
\omega
\right] &= (1-\dleak) \left(\hat{\lchann}^A_j \otimes \hat{\lchann}^B_j\right) \!\left[
\omega
\right] \otimes  
\pure{\gnd}^{\otimes 4}
+ \dleak 
\hat{\lchann}_j \!\left[
\omega
\right],
\end{align}
where the $\pure{\gnd}^{\otimes 4}$ term is on registers $\lk^{A\to B}_j \lk^{A\to E}_j \lk^{B\to A}_j \lk^{B\to E}_j$.\\[-3mm]

\noindent Since this version also seems to be a potentially plausible model to study, we shall analyze it as well in this work. 
Note that a classical-probabilistic leakage constraint implies the bounded-weight version (with the same $\dleak$), but the converse is not generally true (as a simple counterexample, consider a channel $\lchann_j$ that reads the classical value $X_j$ and sets the register $\lk^{A\to B}_j$ to a state of the form $\sqrt{1-\dleak} \ket{\gnd} + \sqrt{\dleak} \ket{\psi_{x_j}}$; a physical example of such states would be weak coherent states in photonic systems). Accordingly, we expect that more ``optimistic'' results could be obtained by assuming a classical-probabilistic leakage constraint, and 
(as we show in our results later) 
we find this can indeed be the case.

However, we now note that when non-IID behaviour is allowed, the above constraints by themselves seem to still be not quite sufficient to allow us to obtain nontrivial security guarantees. This is due to the following attack (that applies for either form of leakage constraint above), which for later reference we shall term a ``random full-leakage attack'': independently in each round, with probability $1-\dleak$ all the leakage registers of that round are set to $\pure{\gnd}$, otherwise both devices set their respective leakage registers to be a copy of \emph{all} the past outputs from that device, which for brevity we shall call a ``full-leakage event''
(note that this event can indeed be coordinated between the devices, because they can hold preshared randomness).
The underlying idea behind this attack is somewhat similar to the device-reuse attacks of~\cite{BCK13}, but we now argue in detail that leakage of this form also poses a problem within a single protocol implementation.

Specifically, recall that if a DIQKD or DIRE protocol is to achieve some nontrivial asymptotic keyrate, that means there is some constant $r>0$ (the asymptotic keyrate) such that the length of the final key approaches $rn$ at large $n$. The preceding attack renders this impossible, by the following argument:
if we take any fraction $r'\in[0,1]$, observe that under that attack, the probability of at least one full-leakage event occurring within the last $r'n$ protocol rounds is $1-(1-\dleak)^{r'n}$, which is close to $1$ at large $n$. Furthermore, a full-leakage event renders the preceding rounds useless, since Eve learns all previous outputs --- hence this means that with probability close to $1$, only at most the last $r'n$ rounds are useful for generating the final key. However, note that if we had picked for instance $r'=0.01r$ (or some appropriately smaller value if the device outputs have high dimension), then those $r'n$ rounds cannot have enough smooth min-entropy (see~\cite{RR12} or Sec.~\ref{sec:fullprot} below) to produce a secret key of length $rn$. 

Note that since the above attack directly obstructs the final goal of secret key generation (with nontrivial asymptotic keyrate), it does not represent a limitation of the proof approaches we use, but rather an inherent limitation of imposing only the preceding constraints on the leakage model.
Hence to obtain nontrivial results in the non-IID case, we shall need to impose further constraints on the leakage registers. Informally, one could say that the main power of the above attack comes from being able to encode a large amount of information in the leakage registers when they are not set to the $\pure{\gnd}$ state. To prevent this, in this work we shall consider two possible approaches: in Sec.~\ref{sec:dimbnd} we consider a dimension bound on the leakage registers, while in Sec.~\ref{sec:Ebnd} we consider a ``softer'' version of a dimension constraint in the form of an energy bound. We elaborate more on those constraints (and how they can be motivated) in their respective sections. 

\begin{remark}\label{remark:nonIID}
It is perhaps worth noting that in the IID case, however, these further constraints are in fact unnecessary --- the preceding attack relies on highly non-IID behaviour, and if one makes an IID assumption, it is possible to obtain nontrivial results without those further constraints. We sketch out the proof structure for that case in Remark~\ref{remark:IIDcase} in Sec.~\ref{sec:fullprot}. (This is also why we defer the discussion of the dimension or energy constraints to Sec.~\ref{sec:fullprot} on the full protocol analysis against non-IID attacks, because we do not in fact need those constraints yet when analyzing single rounds in Sec.~\ref{sec:1rnd}.)

We also note that in device-dependent QKD, a common proof technique for the non-IID case is to use de Finetti arguments~\cite{rennerthesis,CKR09} to reduce the analysis to the IID case. However, once (non-IID) leakage is allowed, there is a subtle obstacle in making this reduction valid even for device-dependent QKD. Specifically, the de Finetti arguments show that (roughly speaking) it suffices to consider the case where the quantum states supplied to Alice and Bob are IID; however, when leakage is allowed, this does \emph{not} imply that the leakage registers are also IID. Hence even for device-dependent QKD, existing arguments do not seem sufficient to reduce the non-IID case to the IID case when leakage is allowed.\footnote{As for device-dependent non-IID security proofs based on complementarity~\cite{SP00,Koa09} or one-shot uncertainty relations~\cite{TL17}, their relation to the IID case is somewhat less straightforward, so the question of whether they achieve a reduction to the IID case in the presence of leakage is somewhat ill-posed.
Although, the latter approach could be compatible with our analysis in Sec.~\ref{sec:fullprot}, since it proceeds via first finding a bound on the smooth min-entropy.}
\end{remark}

For convenience, before proceeding further we list a quick summary of the constraints imposed in our leakage model:
\begin{itemize}
\item Either a bounded-weight or classical-probabilistic leakage constraint for each round
\item Either a dimension bound or energy bound for each round (only needed in a non-IID scenario; see Remark~\ref{remark:IIDcase}) 
\item A ``restricted adaptiveness'' constraint on Eve (only relevant in a non-IID scenario; discussed further in Sec.~\ref{sec:fullprot})
\end{itemize}

\subsection{Security proof structure and protocol requirements}

In this work, we shall focus on the security proofs for DIRE or DIQKD that are based on the \term{entropy accumulation theorem} (EAT)~\cite{DFR20,DF19,LLR+21,NDN+22}.\footnote{A different proof technique known as \term{quantum probability estimation} (QPE)~\cite{KZB20,ZKB18,ZFK20} was used in some DIRE experiments~\cite{SZB+21,LZL+21}, but this approach is sufficiently different that we will not be discussing it in detail here. However, it does basically end up providing a bound on smooth min-entropy similar to that discussed in Sec.~\ref{sec:fullprot} based on the EAT, so the analysis in that section should also generalize to the QPE approach.}
The EAT provides a flexible and modular framework for security analysis of such protocols. 
Roughly speaking, the structure of proofs based on this approach can be divided into two core aspects (see e.g.~\cite{DFR20,TSB+22,LLR+21,NDN+22} for details). The first aspect is focused on analyzing the individual rounds, where the goal is roughly to solve an optimization problem that lower-bounds the von Neumann entropy of the device outputs in a single round (conditioned on some side-information Eve may hold) --- this quantity is useful for characterizing the asymptotic keyrates of such protocols, according to the EAT (or in the simpler IID case, this follows from e.g.~the Devetak-Winter formula~\cite{DW05} or the quantum asymptotic equipartition property (AEP)~\cite{TCR09}). The second aspect is to convert the solution to this optimization into an explicit bound on the length of secure key that can be produced by the protocol after a finite number of rounds. In the subsequent sections, we discuss each of these aspects in more detail, and how they can be modified to account for the leakage model we have described above.

Since in this work we are focusing on EAT-based proofs, we suppose that the protocol has the following structure, to maintain compatibility with the EAT analysis. 
First, we suppose that each of the $n$ rounds is independently chosen with some probability $\gamma$ to be a \term{test round} (informally, one that is used to gather statistics that determine whether the protocol aborts), and is otherwise taken to be a \term{generation round} (informally, one that generates entropy for a raw key that will be processed into a final key). When a test round occurs, Alice and Bob generate values $x$ and $y$ respectively with some fixed probabilities $\ptest_{xy}$ (using trusted local randomness), and supply these values as the inputs to their devices in that round. Analogously, in a generation round, Alice and Bob generate inputs with some other probabilities $\pgen_{xy}$.
After these $n$ input-output pairs have been gathered, some further classical post-processing procedures
are performed to produce the final secret key for the protocol. 
In particular, these procedures include\footnote{Some other post-processing procedures that we do not discuss here include for instance an \term{error correction} step in the case of DIQKD; these procedures does not significantly affect the analysis in this work and hence we do not describe them further.
} a \term{parameter estimation} step, where the protocol accepts if\footnote{The EAT can in fact also accommodate other forms of accept condition, for instance an accept condition based only on the CHSH winning probability, rather than a condition that involves every input-output tuple $(a,b,x,y)$ individually. However, for brevity in this work we will focus on the version described.} (for all input-output tuples $(a,b,x,y)$) the observed frequency of rounds that are test rounds in which Alice and Bob supplied inputs $x,y$ and obtained outputs $a,b$ lies within some small interval\footnote{Informally, this interval is simply a ``tolerance'' parameter to ensure the honest behaviour is still accepted with high probability when accounting for finite-size statistical fluctuations --- see e.g.~\cite{LLR+21,TSB+22} for precise calculations of the required interval widths.} around $\gamma \ptest_{xy} \constr^\mathrm{hon}_{ab|xy}$, where $\constr^\mathrm{hon}_{ab|xy}$ are the probabilities of getting outputs $a,b$ given inputs $x,y$ when the devices are honest.
Furthermore, the last classical post-processing procedure is a \term{privacy amplification} step, in which Alice and Bob process their data in such a way that it ``amplifies'' its privacy with respect to Eve, producing an ideal secret key (see~\cite{rennerthesis} for details). Our subsequent discussion will be based around protocols with this structure (in particular, we remark that this is indeed sufficient to cover the existing finite-size demonstrations of DIRE and DIQKD).

\section{Single-round entropy bounds}
\label{sec:1rnd}

\subsection{Fundamental optimization task}
\label{sec:1rndintro}

We now describe more precisely the relevant single-round optimization problem that has to be analyzed in an EAT-based security proof~\cite{DFR20,TSB+22,NDN+22}. Since we are focusing on a single round, for brevity in this section we shall omit the $j$ subscripts specifying individual rounds. 
Let us first recall the physical process in each round according to our leakage model: after the state preparation process involving the memory registers, there is some quantum state $\omega_{\qA \qB}$ in which Alice and Bob's devices hold registers $\qA$ and $\qB$ respectively. 
Let $\omega_{\qA \qB R}$ be an arbitrary purification of this state (this can be informally thought of as Eve's side-information, though for the purposes of this optimization problem it is just an abstract register holding a purification of $\omega_{\qA \qB}$).
Now, depending on whether it is a test or generation round, the input values $XY$ are generated according to the probabilities $\ptest_{xy}$ or $\pgen_{xy}$, using trusted local randomness. 
After these inputs are supplied to the devices, the leakage channel $\lchann: \qA \qB XY \to \qA \qB \lk^\mathrm{all} XY$ (writing $\lk^\mathrm{all} \defvar \lk^{A\to B} \lk^{A\to E} \lk^{B\to A} \lk^{B\to E}$ for brevity) is then applied. For the test and generation round cases, let us denote the resulting states as $\rho^\mathrm{test}$ and $\rho^\mathrm{gen}$ respectively, so we have
\begin{align}
\begin{gathered}
\rho^\mathrm{test}_{\qA \qB \lk^\mathrm{all} XY R} = \sum_{xy} \ptest_{xy} 
\rho^{xy}_{\qA \qB \lk^\mathrm{all} XY R}, \qquad
\rho^\mathrm{gen}_{\qA \qB \lk^\mathrm{all} XY R} = \sum_{xy} \pgen_{xy} 
\rho^{xy}_{\qA \qB \lk^\mathrm{all} XY R},
\\
\text{where } \rho^{xy}_{\qA \qB \lk^\mathrm{all} XY R} = 
(\lchann \otimes \idmap_R) \left[\omega_{\qA \qB R} \otimes \pure{xy}_{XY}\right].
\end{gathered}
\label{eq:afterleak}
\end{align}
Alice and Bob's devices then perform some unknown measurements $\mchann^A: \qA \lk^{B\to A} X \to A X$ and $\mchann^B: \qB \lk^{A\to B} Y \to B Y$ respectively (in a minor abuse of notation, here for brevity we omit the memory registers from these channel outputs, as they are not involved in analyzing this round). Let us use the following notation for the reduced states on $AB XY R$ after these measurements:
\begin{align}
\begin{gathered}
\rho^\mathrm{test}_{AB XY R} = \sum_{xy} \ptest_{xy} 
\rho^{xy}_{AB XY R}, \qquad
\rho^\mathrm{gen}_{AB XY R} = \sum_{xy} \pgen_{xy} 
\rho^{xy}_{AB XY R}, 
\\
\text{where } \rho^{xy}_{AB XY R} = \left(\mchann^A \otimes \mchann^B \otimes \idmap_{R} \right)\left[\rho^{xy}_{\qA \qB \lk^{B\to A} \lk^{A\to B} XY R}\right].
\end{gathered}
\label{eq:aftermeas}
\end{align}
(The $\rho^\mathrm{test}_{AB XY R}$ state as written above matches $\rho^\mathrm{test}_{\qA \qB \lk^\mathrm{all} XY R}$ on all registers that are present in both expressions, so there is no danger of ambiguity in using $\rho^\mathrm{test}$ to denote both; analogously for $\rho^\mathrm{gen}$ and for each $\rho^{xy}$.) 
Note that these states are classical on the registers $ABXY$. 
Finally, let us extend the states~\eqref{eq:aftermeas} to an additional classical register $S$ that is computed from $ABXY$; this register $S$ represents the value produced from this round that will be incorporated into the raw key of the protocol (we introduce this register for flexibility in our discussion, since roughly speaking in DIRE one typically takes $S=AB$, whereas in DIQKD one often takes $S=A$; see e.g.~\cite{BRC20,TSB+22}).

With this process in mind, the core single-round optimization problem that needs to be solved is as follows:
for any values $\constr_{ab|xy} \in \mathbb{R}$, we need to evaluate or lower-bound the optimization\footnote{Following standard conventions in optimization theory, if the optimization is infeasible then its value is taken as $+\infty$.}
\begin{align}
\label{eq:mainoptavg}
\begin{gathered}
\inf_{\omega, \lchann, \mchann^A, \mchann^B} H(S|XYR)_{\rho^\mathrm{gen}}\\
\suchthat \quad 
\rho^\mathrm{test}_{ABXY} = \sum_{abxy} \ptest_{xy}\constr_{ab|xy} \pure{abxy},
\end{gathered}
\end{align}
where the states $\rho^\mathrm{test},\rho^\mathrm{gen}$ are to be understood as functions of the state $\omega_{\qA \qB R}$ and the channels $\lchann, \mchann^A, \mchann^B$ via~\eqref{eq:afterleak}--\eqref{eq:aftermeas} (and those channels implicitly have the structure described in Sec.~\ref{sec:leakmodel}, including whichever leakage constraint we are considering). On an informal level, 
the objective value $H(S|XYR)_{\rho^\mathrm{gen}}$ roughly characterizes Eve's uncertainty about the raw-key value $S$, conditioned on a side-information register $R$ and the input values $XY$ (which are typically publicly announced at some point in DI protocols).
As for the constraints, the values $\ptest_{xy}\constr_{ab|xy}$ can roughly be thought of as characterizing the 
states $\rho^\mathrm{test}_{ABXY}$
produced by devices that we would accept with high probability in the protocol (for instance, in the IID asymptotic case, we can view them as being the values $\ptest_{xy}\constr^\mathrm{hon}_{ab|xy}$ that would be produced by the honest devices, and suppose that the protocol only accepts devices that exactly reproduce those values). 
However, a detailed discussion of how to convert an algorithm for solving this optimization into a full EAT-based security proof is beyond the scope of this work; we defer the details to e.g.~\cite{DF19,LLR+21,TSB+22,arx_BFF21} (refer to the sections on \term{crossover min-tradeoff functions}). 

We highlight that in the above discussion, the registers $\lk^{A\to E} \lk^{B\to E}$ do not in fact appear in the final optimization~\eqref{eq:mainoptavg}, despite the fact that we are allowing Eve to collect them as side-information. This is because we will be handling those registers separately, in the full-protocol analysis in Sec.~\ref{sec:fullprot}. Still, we note that if one makes an assumption that Eve's attack is IID, then it would indeed be possible to perform a security proof by directly including $\lk^{A\to E} \lk^{B\to E}$ in the conditioning registers in the analysis here. However, without this IID assumption, we would need to rely on other techniques such as the EAT, in which case the registers $\lk^{A\to E} \lk^{B\to E}$ potentially violate a Markov condition~\cite{DFR20,DF19} or no-signalling condition~\cite{arx_MFSR22} in the theorem, and hence it is not useful to compute a bound that already includes them at this step --- we instead handle them separately in the Sec.~\ref{sec:fullprot} analysis.

\begin{remark}\label{remark:IIDcase}
To elaborate further on the IID case: first let us specify that by ``IID attacks'', in this work we mean that in every round Eve independently generates and distributes the same state $\omega_{\qA \qB R}$ across the devices, and the measurement channels $\mchann_j$ and leakage channels $\lchann_j$ are also identical in every round, with the memory registers $\mA_j \mB_j$ and ``update channels'' all being trivial.
With this, the state produced 
after all measurements are performed will be of the form 
$(\gamma \rho^\mathrm{test}_{AB XY  \lk^{A\to E} \lk^{B\to E} R} + (1-\gamma)\rho^\mathrm{gen}_{AB XY  \lk^{A\to E} \lk^{B\to E} R})^{\otimes n}$.
In that case, to compute the keyrate it basically suffices (again, due to the AEP~\cite{TCR09}; see~\cite{TSB+22} for a detailed explanation of the security proof structure) to have a method to evaluate the optimization~\eqref{eq:mainoptavg} except with $H(S|XY \lk^{A\to E} \lk^{B\to E} R)_{\rho^\mathrm{gen}}$ as the objective function, which can be achieved using the same arguments as we have given above. With $\lk^{A\to E} \lk^{B\to E}$ included in the conditioning registers of this optimization (and given the IID structure), it is not necessary to separately handle those registers in the subsequent analysis, i.e.~the Sec.~\ref{sec:fullprot} analysis can be omitted in this scenario.
\end{remark}

\subsection{Relaxing the optimization}
\label{sec:relaxedopt}

The optimization~\eqref{eq:mainoptavg} is not straightforward to solve directly, because the leakage channel $\lchann$ may potentially have some complicated structure. (Furthermore, the optimization over the states \emph{and} measurements makes it a nonconvex problem with systems of unbounded dimension, though techniques have been developed~\cite{NPA08,arx_BFF21} to address these issues in the context of Bell nonlocality and DI cryptography, which we shall also be relying on subsequently.) However, we shall now discuss methods to relax it to more tractable versions, for each of the leakage models discussed in~\ref{sec:leakmodel}.

We begin by first rewriting the optimization~\eqref{eq:mainoptavg} slightly: note that in the constraint, the $\pure{xy}$ terms are orthogonal for distinct $(x,y)$ values, and hence that constraint is equivalent to having an individual constraint for each $(x,y)$ value. Written in the latter form, the factors of $\ptest_{xy}$ and the $\pure{xy}$ terms can be ``cancelled off'' in the constraints, allowing us to write the optimization as
\begin{align}
\label{eq:mainoptxy}
\begin{gathered}
\inf_{\omega, \lchann, \mchann^A, \mchann^B} H(S|XYR)_{\rho^\mathrm{gen}}\\
\suchthat \quad 
\rho^{xy}_{AB} = \sum_{ab} \constr_{ab|xy} \pure{ab} \quad \forall x,y .
\end{gathered}
\end{align}

With this, we now discuss the bounded-weight leakage model. Our key observation for this model is that since we have the constraint that measuring the systems $\lk^{A\to B} \lk^{B\to A}$ (in an appropriate basis) would return the outcome $\pure{\gnd}^{\otimes 2}$ with probability at least $1-\dleak$, we can apply a slight modification of the Gentle Measurement Lemma (see Appendix~\ref{app:GML}) to conclude that the states $\rho^{xy}$ after applying the leakage channel have the following property:
\begin{align} 
F\left(\rho^{xy}_{\qA \lk^{B\to A} \qB \lk^{A\to B} XYR}, \rho^{xy}_{\qA \qB XYR } \otimes \pure{\gnd}^{\otimes 2}\right) \geq 1-\dleak
.
\label{eq:closeness}
\end{align}
In other words, they are in fact ``close'' to the states that would be produced if we simply discard the leakage registers $\lk^{A\to B} \lk^{B\to A}$ and re-initialize them in the fixed state $\pure{\gnd}^{\otimes 2}$. By concavity of fidelity, analogous bounds hold for the states $\rho^\mathrm{test},\rho^\mathrm{gen}$, which are mixtures of the states $\rho^{xy}$. (Here we began by presenting the bound~\eqref{eq:closeness} for each $\rho^{xy}$ rather than the mixtures $\rho^\mathrm{test},\rho^\mathrm{gen}$, because this will allow us to impose sharper constraints in our final optimization --- this is also basically the reason we first rewrote the optimization~\eqref{eq:mainoptavg} in the form~\eqref{eq:mainoptxy}.)

The above property implies that we can lower-bound the optimization~\eqref{eq:mainoptxy} by noting that the states involved in it are ``close'' to some other states where effectively no leakage has occurred, after which we can handle the latter using standard techniques for DI optimizations. 
More precisely, we see that~\eqref{eq:mainoptxy} is lower-bounded by the following optimization, which is written entirely in terms of some states $\sigma$ that are produced without leakage (this result is quite intuitive, but we provide a detailed derivation in Appendix~\ref{app:bndwtopt}):
\begin{align}
\label{eq:bndwtopt}
\begin{gathered}
\inf_{\omega, \widetilde{\mchann}^A, \widetilde{\mchann}^B} H(S|XYR)_{\sigma^\mathrm{gen}} - \fcont(\dleak)\\
\suchthat \quad F\!\left(\sum_{ab} \constr_{ab|xy} \pure{ab}, \sigma^{xy}_{AB}\right) \geq 1-\dleak \quad \forall x,y,
\end{gathered}
\end{align}
where $\widetilde{\mchann}^A, \widetilde{\mchann}^B$ are measurement channels analogous to $\mchann^A, \mchann^B$ except \emph{without} involving leakage registers, and the states $\sigma^\mathrm{gen},\sigma^{xy}$ are defined in terms of these channels and the state $\omega$ in an analogous fashion to~\eqref{eq:afterleak}--\eqref{eq:aftermeas}, omitting the leakage processes (see~\eqref{eq:localmeas}--\eqref{eq:localmeasgen} in Appendix~\ref{app:bndwtopt} for an exact formula).
The function $\fcont$ in the objective is required to be a (uniform) {continuity bound} for conditional entropies, in the following sense: for any states $\rho,\sigma$ with $F(\rho,\sigma) \geq 1-\delta$, we have $\left|H(S|XYR)_{\rho} - H(S|XYR)_{\sigma}\right| \leq \fcont(\delta)$. Such continuity bounds have been a subject of some previous interest~\cite{Win16,SBV+21}; for the purposes of this work, we use the following approach: by the Fuchs--van de Graaf inequalities, the constraint $F(\rho,\sigma) \geq 1-\delta$ implies $d(\rho,\sigma) \leq \sqrt{1-(1-\delta)^2} = \sqrt{2\delta - \delta^2}$, and hence the trace-distance-based continuity bound in~\cite{Win16} lets us take
\begin{align}
\label{eq:fcont}
\fcont(\delta) = t \log\dim(S) + (1+t) \binh\left(\frac{t}{1+t}\right), \text{ where } t=\sqrt{2\delta - \delta^2}.
\end{align}
(The intermediate conversion to trace distance here makes this approach somewhat suboptimal; however, this appears to be the best bound we can obtain using only existing results --- we return to this point at the end of Sec.~\ref{sec:1rnddata} for further discussion.)

The optimization~\eqref{eq:bndwtopt} can now be tackled, because a substantial body of work~\cite{NPA08,NPS14,BSS14,arx_BFF21} in DI cryptography has been developed on the topic of solving optimizations (without leakage) of the form
\begin{align}
\begin{gathered}
\inf_{\omega, \widetilde{\mchann}^A, \widetilde{\mchann}^B} H(S|XYR)_{\sigma^\mathrm{gen}} \\
\suchthat \quad \sigma^{xy}_{AB} = \sum_{ab} \constr_{ab|xy} \pure{ab} \quad \forall x,y,
\end{gathered}
\label{eq:NPAopt}
\end{align}
where the channels $\widetilde{\mchann}^A, \widetilde{\mchann}^B$ and states $\omega,\sigma^{xy}$ have the same structure as in~\eqref{eq:bndwtopt}.
In particular, the techniques in~\cite{NPS14,BSS14,arx_BFF21} transform the objective function in the above optimization problem in such a way that the optimization over states and measurements can be lower-bounded using semidefinite programming (SDP) bounds developed in~\cite{NPA08}. 
Compared to our optimization~\eqref{eq:bndwtopt}, the only significant difference is that we have a looser fidelity-based constraint instead of an exact equality constraint (the $\fcont(\dleak)$ term is simply a constant in the optimization and hence does not introduce any difficulties). Fortunately, this fidelity-based constraint indeed has an SDP formulation (see Appendix~\ref{app:fidSDP}), and hence we can apply the SDP techniques in~\cite{NPS14,BSS14,arx_BFF21} to evaluate our optimization~\eqref{eq:bndwtopt}.

As for the classical-probabilistic leakage model, the analysis is easier: by similar arguments to the above, note that the channel structure~\eqref{eq:classprob} implies we can write (for each $(x,y)$ value)
\begin{align}
\rho^{xy}_{ABXYR} = (1-\dleak) \sigma^{xy}_{ABXYR} + \dleak \hat{\sigma}^{xy}_{ABXYR},
\label{eq:mixture}
\end{align}
for some states $\sigma^{xy}$ produced without leakage (more precisely, of the form presented in~\eqref{eq:localmeas}) and some other uncharacterized\footnote{Strictly speaking, the channel structure $\lchann = \lchann^A \otimes \lchann^B$ imposes some constraints on the channel $\hat{\lchann}$ in~\eqref{eq:classprob}, which could in principle slightly constrain the ``uncharacterized'' states $\hat{\sigma}^{xy}_{ABXYR}$ (and thus the values $\hat{\constr}_{ab|xy}$ in the optimization~\eqref{eq:clpropt} below). However, we will not attempt to analyze this in detail here.}  states $\hat{\sigma}^{xy}$. 
Again, a similar decomposition holds for the states $\rho^\mathrm{test},\rho^\mathrm{gen}$ as well, and by concavity of conditional entropy, the decomposition for the latter then implies we have $H(S|XYR)_{\rho^\mathrm{gen}} \geq (1-\dleak)H(S|XYR)_{\sigma^\mathrm{gen}}$. Putting things together, we can thus relax the optimization~\eqref{eq:mainoptxy} in this case\footnote{Here it does not really matter whether we consider the original optimization~\eqref{eq:mainoptavg} or the rewritten version~\eqref{eq:mainoptxy} that has individual constraints for each $(x,y)$; both approaches yield equivalent results in this case because the constraints in~\eqref{eq:clpropt} are affine and the states $\pure{xy}_{XY}$ are orthogonal.} to
\begin{align}
\label{eq:clpropt}
\begin{gathered}
\inf_{\omega, \widetilde{\mchann}^A, \widetilde{\mchann}^B, \hat{\constr}_{ab|xy} \in \mathbb{R}_{\geq 0}} (1-\dleak) H(S|XYR)_{\sigma^\mathrm{gen}} \\
\begin{aligned}
\suchthat \quad & (1-\dleak) \sigma^{xy}_{AB} 
= \sum_{ab} (\constr_{ab|xy}-\dleak \hat{\constr}_{ab|xy}) \pure{ab}
\text{\; and \;}
\sum_{ab} \hat{\constr}_{ab|xy} = 1 \quad \forall x,y,
\end{aligned}
\end{gathered}
\end{align}
where the states $\sigma^\mathrm{gen},\sigma^{xy}$ are again defined via~\eqref{eq:localmeas}--\eqref{eq:localmeasgen}, and $\hat{\constr}_{ab|xy}$ are some non-negative scalars (basically, the constraints simply arise from writing out the classical state $\hat{\sigma}^\mathrm{test}_{ABXY}$ in its explicit form $\hat{\sigma}^\mathrm{test}_{ABXY} = \sum_{abxy} \ptest_{xy}\hat{\constr}_{ab|xy} \pure{abxy}$ for some unknown conditional probabilities $\hat{\constr}_{ab|xy}$).

\subsection{Numerical results}
\label{sec:1rnddata}

As a small demonstration of our method, and to get a qualitative sense of how the value of $\dleak$ affects the single-round entropies, we now evaluate the optimizations~\eqref{eq:bndwtopt} and~\eqref{eq:clpropt} for some choices of the constraint values $\constr_{ab|xy}$. This is meant to be just a simple example; we do not aim to address the full range of implementations that have been used in DI protocol demonstrations, since in any case the appropriate choice of $\constr_{ab|xy}$ for a given experiment would depend on the details of the implementation. (Our approach does not place any particular requirements on $\constr_{ab|xy}$, so it should generically be usable for any choice of those values, as long as the input and output sizes are not so large that the~\cite{NPA08} SDPs become intractable.)

Specifically, we follow e.g.~\cite{PAB+09,ARV19} and consider distributions $\constr_{ab|xy}$ with binary-valued outputs $a,b\in\{0,1\}$ that are produced by measurements on a \term{Werner state}, parametrized by a depolarizing-noise value $\q\in[0,1/2]$ (this parametrization arises from the fact that if matching Pauli measurements in the $X$-$Z$ plane of the Bloch sphere are performed on the qubits in the Werner state, then the probability of getting different outcomes is simply $\q$):
\begin{align}
\rho^\mathrm{W}_{\q} \defvar (1-2\q)\pure{\Phi^+} + 2\q\, \frac{\id}{4}, \quad \text{where } \ket{\Phi^+} \defvar \frac{1}{\sqrt{2}} \left(\ket{00}+\ket{11}\right).
\label{eq:werner}
\end{align}
We consider two choices of measurements on this state, with different numbers of inputs (i.e.~measurement settings) in each. Firstly, we consider a scenario with $4$ inputs each for Alice and Bob, 
where each input value corresponds to the same measurement by either Alice or Bob, namely a (rotated) Pauli measurement in the $X$--$Z$ plane at the following polar angles $\theta$:
\begin{align}
\begin{gathered}
x=0: \theta = 0, \qquad x=1: \theta = \frac{\pi}{4}, \qquad x=2: \theta = \frac{\pi}{2}, \qquad x=3: \theta = \frac{3\pi}{4}, \\
y=0: \theta = 0, \qquad y=1: \theta = \frac{\pi}{4}, \qquad y=2: \theta = \frac{\pi}{2}, \qquad y=3: \theta = \frac{3\pi}{4}.
\end{gathered}
\label{eq:measMYCHSH}
\end{align}
This can be viewed as a combination of the Mayers-Yao self-test~\cite{MY98} with the measurements that maximize the violation of the CHSH inequality (up to sign conventions)~\cite{CHSH69}, or alternatively as having both Alice and Bob perform the latter measurements. Secondly, we consider a simpler situation with just $2$ inputs each for Alice and Bob, corresponding to (rotated) Pauli measurements at the following angles in the $X$--$Z$ plane:
\begin{align}
\begin{gathered}
x=0: \theta = 0, \qquad x=1: \theta = \frac{\pi}{2}, \\
y=0: \theta = \frac{\pi}{4}, \qquad y=1: \theta = \frac{3\pi}{4}.
\end{gathered}
\label{eq:measCHSH}
\end{align}
These measurements are the ones that maximize the violation of the CHSH inequality, and are very commonly studied in DI security proofs.

As for the choices of the testing probabilities $\ptest_{xy}$ and $\pgen_{xy}$, in both scenarios we take $\ptest_{xy}$ to be uniform over all inputs, and we set $\pgen_{00}=1$ and $\pgen_{xy}=0$ otherwise, i.e.~the only inputs we use in the generation rounds are $x=y=0$. We also set the register $S$ to just be equal to $A$, i.e.~the relevant quantity in DIQKD protocols. With these choices, the objective function can be written as $H(A|R;X=Y=0)$. 
We plot our resulting bounds on the optimizations~\eqref{eq:bndwtopt},~\eqref{eq:clpropt} in Figs.~\ref{fig:MYCHSH}--\ref{fig:CHSH}.
Regarding the details of the numerical computations: we used the SDP approaches from~\cite{NPS14,BSS14} which lower-bound $H(A|R;X=Y=0)$ in terms of guessing probability, and for the 4-input scenario~\eqref{eq:measMYCHSH} we used local level 1 of the~\cite{NPA08} hierarchy, while for the 2-input scenario~\eqref{eq:measCHSH} we used local level 2 of the~\cite{NPA08} hierarchy.

\begin{figure}[t]
\centering
\subfloat[]{
\includegraphics[width=0.49\textwidth]{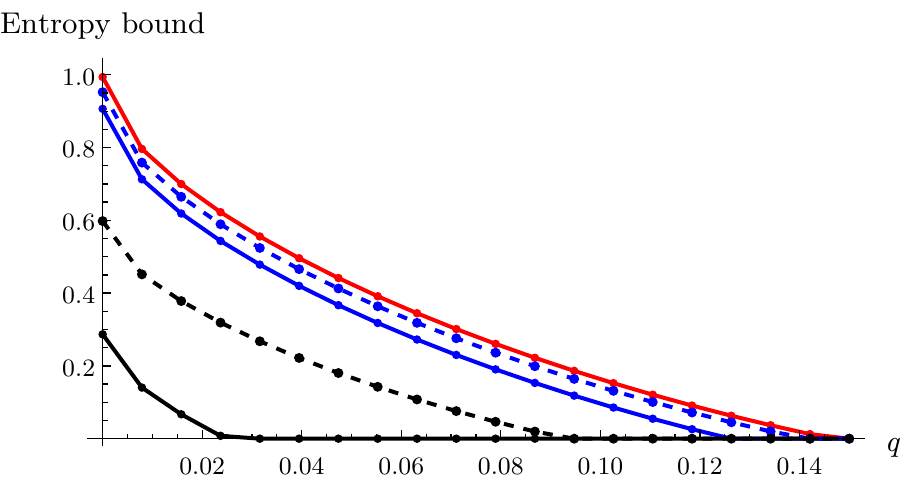}
} 
\subfloat[]{
\includegraphics[width=0.49\textwidth]{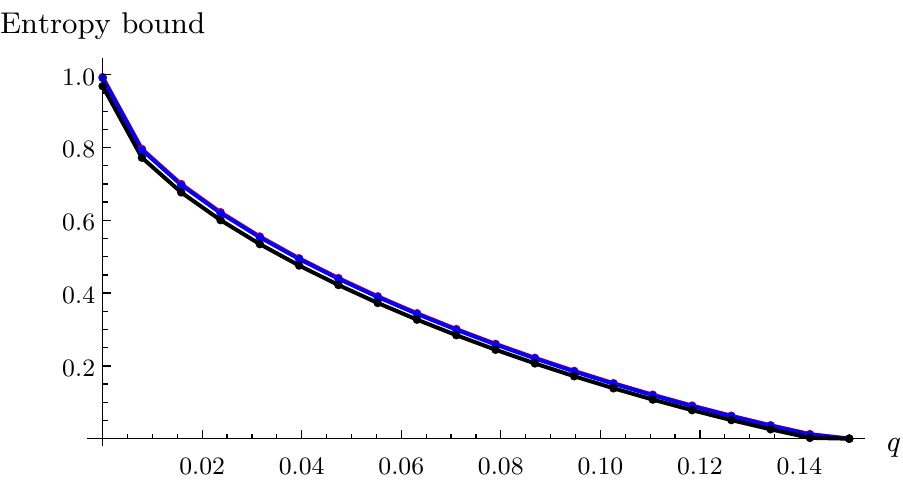}
}
\caption{
Lower bounds on $H(A|R;X=Y=0)$ (via guessing probability) for a 4-input 2-output scenario given by
the measurements~\eqref{eq:measMYCHSH} on a Werner state with
depolarizing noise $\q$.
In each plot the solid black, blue and red curves respectively denote the results for (a) the bounded-weight leakage model with $\dleak=10^{-3},10^{-5},0$; (b) the classical-probabilistic leakage model with $\dleak=10^{-2},10^{-3},0$ (note that we have chosen quite different orders of magnitude for $\dleak$ in the two models, because the latter is significantly more robust against leakage, as can be seen from the plots). For the bounded-weight model, we have also plotted dashed curves that show the results if the continuity-bound term $\fcont(\dleak)$ is omitted from the optimization~\eqref{eq:bndwtopt}; it can be seen that this term has a very significant impact.
}
\label{fig:MYCHSH}
\vspace{5mm}
\centering
\subfloat[]{
\includegraphics[width=0.49\textwidth]{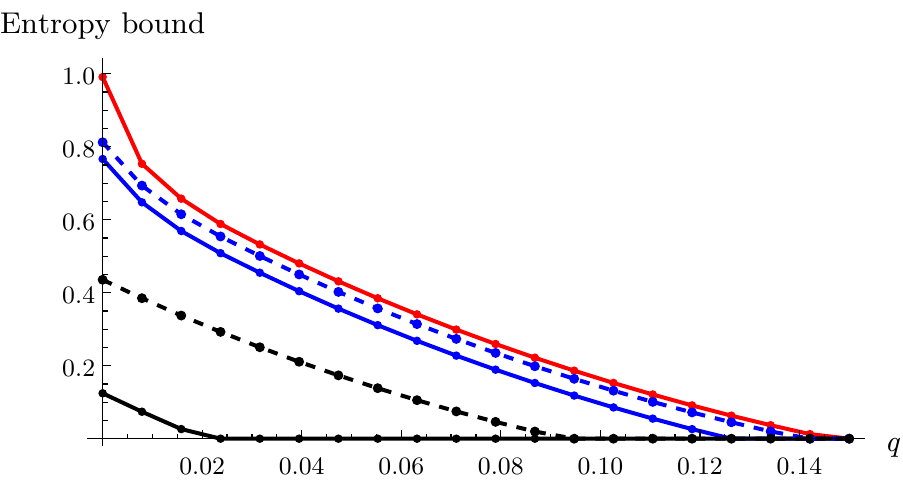}
} 
\subfloat[]{
\includegraphics[width=0.49\textwidth]{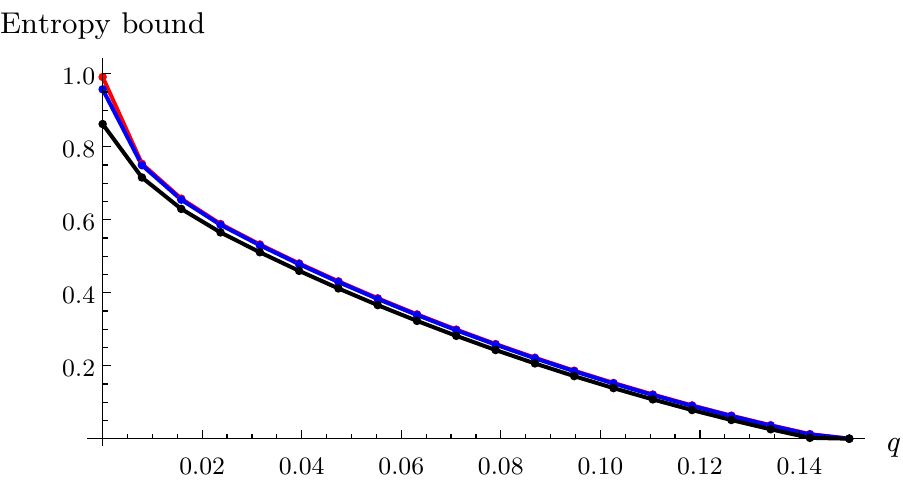}
}
\caption{
Analogous to Fig.~\ref{fig:MYCHSH}, except for a 2-input 2-output scenario given by
the measurements~\eqref{eq:measCHSH}, implemented on a Werner state with
depolarizing noise $\q$ as well.
}
\label{fig:CHSH}
\end{figure}

From the plots, it can be seen that our entropy bounds for the bounded-weight leakage model are lower than those for the classical-probabilistic leakage model (given the same $\dleak$ value). This is as expected since the latter is a special case of the former; however, it is noteworthy that the difference turned out to be quite dramatic --- for instance at $\dleak=10^{-3}$ the entropy in the bounded-weight model is already very low, while that in the classical-probabilistic model is only slightly affected. This indicates that under the bounded-weight model, we cannot tolerate particularly large values of $\dleak$ before the results become trivial, at least with the methods proposed here. (Similar behaviour was observed in the analysis of prepare-and-measure scenarios in~\cite{PPWT22}, though the analysis in that work was mostly based on trace distance rather than fidelity.)

Still, from the plots it can also be seen that the continuity-bound term $\fcont(\dleak)$ is playing a significant role in reducing the entropy in the bounded-weight model. This is likely because the approach we have used here to obtain the continuity bound $\fcont(\dleak)$ is rather suboptimal --- firstly, despite originally having a bound on the fidelity, we took an intermediate step of converting it to trace distance; secondly, working with the trace distance has an inherent disadvantage that the von Neumann entropy is not Lipschitz continuous with respect to trace distance (i.e.~viewing the formula~\eqref{eq:fcont} as a function of trace distance $t$, it grows faster than any linear function at small values of $t$, and this scaling behaviour is essentially unavoidable~\cite{Win16}). These points combined led to a continuity bound $\fcont(\dleak)$ that scales approximately as $O(\binh(\sqrt{2\dleak}))$ (taking the dominant terms in~\eqref{eq:fcont} at small $\dleak$).
If we had directly worked with the lower bound of $1-\dleak$ on the fidelity, some results in~\cite{SBV+21,arx_LKA+22} heuristically suggest that it might be possible to instead obtain a continuity bound that is Lipschitz with respect to the angular distance $\acos(1-\dleak)$, i.e.~$\fcont(\dleak)$ would scale as $O(\acos(1-\dleak)) \approx O(\sqrt{2\dleak})$. This is significantly better than the $O(\binh(\sqrt{2\dleak}))$ scaling in~\eqref{eq:fcont}, and would lead to tighter bounds. However, it has not been rigorously proven that such a Lipschitz continuity bound (with respect to angular distance) holds when the conditioning systems are quantum; 
the results here indicate that deriving such a result in future work could be very useful in improving the keyrates from this approach.

\section{Security of full protocol}
\label{sec:fullprot}

We now turn to the question of ensuring security of the entire protocol. Informally, the idea here is to argue that with the constraints on the leakage systems, they cannot reveal ``too much'' additional information to Eve as compared to a situation where no leakage occurs. However, formalizing this intuition turns out to be rather technical, and we begin by laying out some of the required foundations. 

Recall we have the requirement that the protocol ends in a privacy amplification step, where the parties process their data (after all the other classical post-processing steps) into an ideal secret key. Let $\str{S}$ denote the string to be processed in that step, and let $E_\mathrm{all}$ denote all the side-information that Eve holds at that point. It is known that the length of secret key that can be produced through privacy amplification is essentially characterized~\cite{RW05}\footnote{Recent work~\cite{arx_Dup21} has provided an alternative characterization in terms of {{\Renyi} entropies}, but we defer further discussion of this version to the {conclusion}.} by the (conditional) {smooth min-entropy} $\Hmin^{\es}\left(\str{S}| E_\mathrm{all} \right)$ (Definition~\ref{def:Hsmooth}). 
(To be precise, in a full security proof this entropy should be evaluated for the state {conditioned} on the protocol accepting, and there are subtleties in choosing the security definitions in a manner compatible with rigorous analysis of this conditioning. However, a detailed discussion is beyond the scope of this work; refer to e.g.~\cite{TL17,arx_tanthesis,arx_PR21} for further discussion.)

In the standard DI scenario where there is no leakage, a security proof can typically proceed by taking $E_\mathrm{all}$ to consist of two parts: the register $\En$ representing the side-information Eve holds immediately after Alice and Bob have collected all their device outputs and announced their input choices (see Appendix~\ref{app:EAT} for a discussion of some technical details on this point), 
and a register $\str{P}$ holding all the public communication after that point. (In the case of DIRE, $\str{P}$ is typically small or trivial, but we include it here to maintain generality for DIQKD.)
The main contribution of the EAT is that it provides a lower bound on $\Hmin^{\es}\left(\str{S}| \En \right)$ (given a method to bound the single-round optimization~\eqref{eq:mainoptavg}).
This is then converted into a lower bound on the quantity of interest $\Hmin^{\es}\left(\str{S}| E_\mathrm{all} \right) = \Hmin^{\es}\left(\str{S}| \En \str{P} \right)$, by simply lower bounding it with $\Hmin^{\es}\left(\str{S}| \En \right) - \log\dim(\str{P} )$ using a chain rule~\cite{WTHR11,Tom16} 
for smooth min-entropy. 

In our context, it would seem we could modify the above approach by simply appending the string of leakage registers $\str{\lk}^{A\to E} \str{\lk}^{B\to E}$ (collected by Eve) to $E_\mathrm{all}$, and aim to find a way to relate $\Hmin^{\es}\left(\str{S}| \str{\lk}^{A\to E} \str{\lk}^{B\to E} \En \str{P} \right)$ to $\Hmin^{\es}\left(\str{S}| \En \str{P} \right)$ (since by the above outline, an EAT-based security proof already provides methods to bound the latter). However, this runs into a subtle difficulty --- since the leakage registers $\lk^{A\to E}_j \lk^{B\to E}_j$ in each round are immediately leaked to Eve, in principle she could perform a joint operation on these registers and the side-information she holds at that point, in order to generate the states sent to Alice and Bob's devices in the next round. Such an operation could potentially couple the registers $\lk^{A\to E}_j \lk^{B\to E}_j$ to $\qA_{j+1} \qB_{j+1}$ in some complicated way, which prevents one from applying the EAT to bound $\Hmin^{\es}\left(\str{S}| \En \str{P} \right)$ 
(see Appendix~\ref{app:EAT} for a specialized discussion of the details).
More trivially, it could also change the state on the registers $\lk^{A\to E}_j \lk^{B\to E}_j$ in some uncharacterized fashion, making it hard to preserve any particular structure for the state on those registers at the end.

Hence in this work, we impose a form of ``restricted adaptiveness'' assumption on Eve, as was described in the leakage model in Sec.~\ref{sec:leakmodel}.
Specifically, we have supposed in that model that Eve's strategy consists of gathering the leakage registers $\lk^{A\to E}_j \lk^{B\to E}_j$ in each round, but not operating further on them until all the state distribution and measurement steps have been completed. 
We note that this condition would indeed be plausible in, for instance, a DIRE implementation where the entire process of state generation and measurement takes place within a ``shielded'' lab where the leakage between the devices and out of the lab can be constrained (but apart from this constraint, the honest parties do not have any \emph{a priori} certification of the states or measurements), and any other side-information Eve can hold about the measurement outcomes must come from some (possibly quantum) extension that she kept before distributing the devices. Such a context has been considered in for instance~\cite{PM13} (though only for zero leakage rather than bounded leakage, and focused on classical side-information), and is similar to the usage contexts for existing QRNG devices (though those are based on characterized states and/or measurements). On the other hand, in the context of DIQKD it is perhaps harder to justify this assumption, since in this setting one usually considers Alice and Bob to be receiving their states from an untrusted source. Still, we note that if an adaptive attack for Eve can be modelled via some state preparation process using the memory registers $\mA_j \mB_j$ described in Sec.~\ref{sec:leakmodel}, for instance if there is a way to ``partition off'' the parts of the leakage that encode the adaptive behaviour and include them in $\mA_j \mB_j$ rather than $\lk^{A\to E}_j \lk^{B\to E}_j$, then our model would be sufficient to cover it as well. 

With this restriction, 
our task is then to bound 
the smooth min-entropy of $\str{S}$ conditioned on $\str{\lk}^{A\to E} \str{\lk}^{B\to E} \En \str{P}$, 
where the registers $\str{\lk}^{A\to E} \str{\lk}^{B\to E}$ were not modified after they were initially generated. Intuitively, we would expect that it should be possible to compensate for the leakage registers by subtracting some amount from $\Hmin^{\es}\left(\str{S}|\En \str{P} \right)$ (which can be bounded using the EAT-based analysis as described above). However, the simple dimension-based chain rule that was used to handle the $\str{P}$ register is not sufficient in this context, because the log-dimension of $\str{\lk}^{A\to E} \str{\lk}^{B\to E}$ could easily be larger than $\Hmin^{\es}\left(\str{S}|\En \str{P} \right)$ (for instance simply if each $\lk^{A\to E}_j$ and/or $\lk^{B\to E}_j$ register is the same dimension as $S_j$). Instead, we use some different chain rules~\cite{VDT13,Tom16} involving the {smooth max-entropy} (Definition~\ref{def:Hsmooth}).
Specifically, introducing the notation
\begin{align}
\smf{p}\defvar 
-\log\left(1-\sqrt{1-p^2}\right)
\leq \log
\frac{2}{p^2}
,
\end{align}
if we take any values $\es,\esp,\tau,\esL \in (0,1)$ such that $\esp=\es+2\tau+4\esL$, the following inequality holds (for any state on some registers $QQ'Q''$):
\begin{align}
\Hmin^{\esp}(Q|Q'Q'') &\geq \Hmin^{\es+\tau+2\esL}(QQ'|Q'') - \Hmax^{\esL}(Q'|Q'') - 2\smf{\tau} 
&\text{(Eq.~(6.56) of~\cite{Tom16})} 
\nonumber
\\
&\geq \Hmin^{\es}(Q|Q'') + \Hmin^{\esL}(Q'|QQ'') - \Hmax^{\esL}( Q'|Q'') - 3\smf{\tau} 
&\text{(Eq.~(6.55) of~\cite{Tom16})} 
\nonumber
\\
&\geq \Hmin^{\es}(Q|Q'') - 2\Hmax^{\esL}( Q') - 3\smf{\tau} , & \label{eq:chainHminHmax}
\end{align}
where the last line holds by applying a duality relation $\Hmin^{\esL}(Q'|QQ'') = -\Hmax^{\esL}(Q'|Q''')$~\cite{Tom16} (where $Q'''$ is any purification of $QQ'Q''$) and a data-processing property $\Hmax^{\esL}(Q'|R) \leq \Hmax^{\esL}(Q')$. (In principle one could of course choose the parameters
in the two chain rules separately,
arriving at a final bound $\Hmin^{\esp}(Q|Q'Q'') \geq \Hmin^{\es}(Q|Q'') - \Hmax^{\esL}(Q') - \Hmax^{\esL'}( Q')- 2\smf{\tau} - \smf{\tau'} $ for $\esp=\es+\tau+\tau'+2\esL+2\esL'$, but it seems unclear if this provides any benefit --- it essentially comes down to whether the value of $\Hmax^{\esL}(Q') + \Hmax^{\esL'}( Q')$ subject to an upper bound on $\esL+\esL'$ is minimized by taking $\esL=\esL'$, and analogously for $2\smf{\tau} + \smf{\tau'}$.)

\begin{remark}
If the system $Q'$ is classical, then a sharper result is possible since in that case we have $\Hmin^{\es}(QQ'|Q'') \geq \Hmin^{\es}(Q|Q'')$ (Lemma~6.7 of~\cite{Tom16}), 
so we can basically stop after the first inequality in~\eqref{eq:chainHminHmax} and end up subtracting only about $\Hmax^{\esL}( Q')$ instead of $2\Hmax^{\esL}( Q')$.
\end{remark}

\newcommand{\rhon}{\rho_n}
\newcommand{\rhoPE}{\rho_{|\mathrm{PE}}}
\newcommand{\pPE}{p_{\mathrm{PE}}}
\newcommand{\epsPE}{\eps_{\mathrm{PE}}}
In our context, let $\rhon$ denote the state just before privacy amplification, and let $\rhoPE$ denote that state conditioned (with normalization) on the event that the protocol accepted during the parameter-estimation step.\footnote{As briefly mentioned at the start of this section, when performing a security proof what we finally need would instead be the smooth min-entropy of the state conditioned on \emph{all} steps of the protocol accepting. However, we will not further discuss here how to handle conditioning on any additional events; techniques to handle this are described in e.g.~\cite{TL17} (Lemma~10) or \cite{TSB+22} (Sec.~4.2).} By applying the above result twice\footnote{Alternatively, we could just apply it a single time, identifying all the registers $\str{\lk}^{A\to E} \str{\lk}^{B\to E}$ with $Q'$. Our subsequent analysis should also essentially be able to provide upper bounds on $\Hmax^{\esL}\left(\str{\lk}^{A\to E} \str{\lk}^{B\to E}\right)$ with appropriate modifications (e.g.~the choice of value for the dimension bound in Sec.~\ref{sec:dimbnd} may change, or the Hamiltonian used in Sec.~\ref{sec:Ebnd}). 
Whether this alternative approach yields better results would seem to depend on the details of the protocol setup and parameter choices, so we do not discuss it in further depth within this work.
}, we see that for any $\es, \esp,\tau,\esL \in (0,1)$ such that $\esp=\es+4\tau+8\esL$ (again, it would be possible to choose different smoothing parameters 
in each use of the bound,
but we omit this here for brevity), we have
\begin{align}\label{eq:chainleak}
\Hmin^{\esp}\left(\str{S}|\str{\lk}^{A\to E} \str{\lk}^{B\to E} \En \str{P} \right)_{\rhoPE} 
\geq \Hmin^{\es}(\str{S}|\En \str{P} )_{\rhoPE} - 2\Hmax^{\esL}\left(\str{\lk}^{A\to E}\right)_{\rhoPE} - 2\Hmax^{\esL}\left(\str{\lk}^{B\to E}\right)_{\rhoPE} - 6\smf{\tau}.
\end{align}
In other words, this basically means that we can compensate for the registers $\str{\lk}^{A\to E} \str{\lk}^{B\to E}$ by subtracting (twice) their smooth max-entropies 
along with an additional constant ``correction'' term $6\smf{\tau}$ (as well as slightly changing the smoothing parameter, which slightly worsens the final secrecy properties of the key, but not too much --- see the Leftover Hashing Lemma in e.g.~\cite{rennerthesis,TL17} for details). 
Since the EAT can be used to prove that the $\Hmin^{\es}(\str{S}|\En \str{P} )_{\rhoPE}$ term is of order $\Omega(n)$ in typical protocols, the $6\smf{\tau}$ term is an almost-negligible correction, and we do not consider it in further detail.
Our task is hence reduced to upper-bounding the smooth max-entropies $\Hmax^{\esL}\left(\str{\lk}^{A\to E}\right)_{\rhoPE}$ and $\Hmax^{\esL}\left(\str{\lk}^{B\to E}\right)_{\rhoPE}$. Since these terms have basically the same structure in our model, in the remainder of this section we shall for brevity use $\str{\lkE}$ to denote either $\str{\lk}^{A\to E}$ or $\str{\lk}^{B\to E}$, and describe how to bound $\Hmax^{\esL}\left(\str{\lkE}\right)_{\rhoPE}$.

The approach we shall use to bound this smooth max-entropy is to relate it to the {\Renyi} entropies (Definition~\ref{def:Renyi}), using some intermediate results that were proven in the derivation of the EAT in~\cite{DFR20}, as well as some security proof techniques used in e.g.~\cite{rennerthesis,ARV19,TSB+22}.
Specifically, let $\pPE$ be the (unknown) probability that the state was accepted during parameter estimation. Then for any value $\epsPE\in(0,1)$, one of the following must be true:
\begin{itemize}
\item Either $\pPE \leq \epsPE$, in which case the protocol's security condition is trivially satisfied (see e.g.~\cite{ARV19,TSB+22} for details --- the more precise claim would be that the secrecy condition holds with secrecy parameter $\epsPE$) and hence we do not discuss it further here; 
\item Or $\pPE > \epsPE$, in which case for any $\alpha\in[1/2,1)$ we have
\begin{align}
\Hmax^{\esL}\left(\str{\lkE}\right)_{\rhoPE} &\leq H_\alpha(\str{\lkE})_{\rhoPE} +
\frac{\smf{\esL}}{1/\alpha - 1} 
\nonumber\\
&\leq H_\alpha(\str{\lkE})_{\rhon} +
\frac{\log(1/\epsPE)}{1/\alpha - 1} +
\frac{\smf{\esL}}{1/\alpha - 1}
\nonumber\\
&\leq \sum_{j=1}^n \sup_\omega H_\alpha(\lkE_j)_{\lchann_j[\omega]}
+ \frac{\log(1/\epsPE)+\smf{\esL}}{1/\alpha - 1}, \label{eq:chainRenyi}
\end{align}
where we have used the following results from~\cite{DFR20}: the first line is Lemma~B.10 
(which holds for $\alpha\in[1/2,1)$),
the second line is Lemma~B.5\footnote{It would also have been possible to apply Lemma~B.6 instead, but that would yield slightly worse dependence on the Renyi parameter in this context.}
(which holds for $\alpha\in(0,1)$) together with the bound $\pPE > \epsPE$,
and the third line follows from a chain rule for {\Renyi} entropies\footnote{We cannot simply claim that $H_\alpha(\str{\lkE})$ is upper bounded by $\sum_{j=1}^n H_\alpha(\lkE_j)$ (without the supremum over input states to the channels), because the {\Renyi} entropies are not subadditive.} presented as Corollary~3.5 in that work (see also Appendix~\ref{app:EAT} below), with the supremum in the last expression taking place over all input states to the channel $\lchann_j$.
\end{itemize}
For the purposes of our analysis, $\alpha\in[1/2,1)$ can be considered a free parameter that should be chosen to optimize the upper bound on $\Hmax^{\esL}\left(\str{\lkE}\right)_{\rhoPE}$. 
To estimate the scaling of this bound at large $n$, we can follow~\cite{DFR20} and take $1-\alpha \propto 1/\sqrt{n}$, so we have $1/({1/\alpha - 1}) 
\leq 1/({1-\alpha}) 
= O(\sqrt{n})$, and for the dimension-bounded case\footnote{This analysis does not carry over to the energy-bounded case with no dimension bound, because the $O(1-\alpha)$ continuity bound presented in~\cite{DFR20} has a dependence on $\dim(\lkE_j)$ (or at least the dimension of the support of the state, due to a $H_0$ term).
In fact there exist infinite-dimensional
states with infinite {\Renyi} entropy for all $\alpha<1$ but finite von Neumann entropy; whether such states can be ruled out would depend on the Hamiltonian describing a given implementation.
} each $H_\alpha(\lkE_j)$ term converges to $H(\lkE_j)$ on order $O(1-\alpha) = O(1/\sqrt{n})$, by the continuity bounds in~\cite{DFR20}. Accounting for the sum over $j$, this means our upper bound in this case scales as $
\left(\sum_{j=1}^n \sup_\omega H(\lkE_j)_{\lchann_j[\omega]}\right) + O(\sqrt{n})$, roughly similar to the AEP~\cite{TCR09}. (In fact we could also have obtained a bound with this scaling for the dimension-bounded case by directly applying the final EAT result~\cite{DFR20}; however, the approach we use here yields tighter results as we shall analyze $H_\alpha(\lkE_j)$ directly rather than taking an intermediate step of relating it to $H(\lkE_j)$.)

With the bound~\eqref{eq:chainRenyi}, our task is reduced to upper-bounding $\sup_\omega H_\alpha(\lkE_j)_{\lchann_j[\omega]}$ (for all $j$).
At first glance, it might seem that this is possible using only the bounded-weight leakage constraint (or the classical-probabilistic leakage constraint), since that enforces that the state produced by the leakage channel $\lchann_j$ is ``close'' to the pure state $\pure{\gnd}$, which has zero entropy. However, this alone does not quite work, since if the $\lkE_j$ systems can have arbitrarily high dimension, they can still have arbitrarily high $H_\alpha(\lkE_j)$ despite the bounded-weight or classical-probabilistic leakage constraint (as can be seen from the subsequent sections, and also basically implied by the ``random full-leakage attack'' previously described in Sec.~\ref{sec:leakmodel}). Hence in this section we shall also require some choice of additional constraint as mentioned in Sec.~\ref{sec:leakmodel}, i.e.~either a dimension bound or an energy bound.

Before proceeding, we show a helpful reduction to classical probability distributions (rather than quantum states): for each $\lkE_j$ register, let 
$\big\{\ket{e_k}_{\lkE_j}\big\}$ 
be any orthonormal basis such that the first basis state $\ket{e_0}$ is equal to $\ket{\gnd}$, and let $\mathcal{P}$ denote the pinching channel with respect to that basis, i.e.~$\mathcal{P}[\rho] \defvar \sum_k \pure{e_k} \rho \pure{e_k}$. Since pinching channels are unital, they cannot decrease the {\Renyi} entropy~\cite{Tom16}, i.e.~for any state $\rho$ we have 
\begin{align}
H_\alpha(\lkE_j)_{\rho} \leq H_\alpha(\lkE_j)_{\mathcal{P}[\rho]} = \frac{1}{1-\alpha} \log \sum_k w_k^\alpha,
\label{eq:pinchbnd}
\end{align}
where $w_k \defvar \bra{e_k}\rho\ket{e_k}$ (these values $\mathbf{w}$ form a probability distribution, i.e.~we have $w_k \geq 0$ and $\sum_k w_k = 1$).
Since the right-hand-side of the above bound is monotone increasing with respect to the $\sum_k w_k^\alpha$ term, to upper bound $H_\alpha(\lkE_j)$ it suffices to just consider the latter instead.
Furthermore, note that under either the bounded-weight or classical-probabilistic leakage constraint (both these models give the same results in our subsequent analysis, in contrast to Sec.~\ref{sec:1rnd}), we have $w_0 = \bra{\gnd}\rho\ket{\gnd} \geq 1-\dleak$. 
With this, our task is reduced to studying the following optimization: 
\begin{align}
\label{eq:Renyiopt}
\begin{gathered}
\sup_{\mathbf{w} \in \dom} 
\sum_k w_k^\alpha\\
\begin{aligned}
\suchthat \quad & w_0 \geq 1-\dleak,
\end{aligned}
\end{gathered}
\end{align}
where the optimization domain $\dom$ encodes the requirement that $\mathbf{w}$ is a probability distribution, as well as either a dimension bound or an energy bound, depending on which we choose to use (the above optimization is unbounded if $\dom$ is allowed to be e.g.~all probability distributions $\mathbf{w}$ of arbitrary finite dimension). 
Given some upper bound $U_\alpha$ on the above optimization, the quantity of interest $\sup_\omega H_\alpha(\lkE_j)_{\lchann_j[\omega]}$ is simply upper bounded by $\frac{1}{1-\alpha} \log U_\alpha$.

This is now just an optimization of {\Renyi} entropy (with the logarithm omitted) over classical probability distributions; furthermore, the objective function is a concave function of $\mathbf{w}$ 
(recalling that we are using $\alpha<1$) and hence this is a concave optimization as long as $\dom$ is a convex set.
We remark that this approach yields a tight bound on $\sup_\omega H_\alpha(\lkE_j)_{\lchann_j[\omega]}$ whenever for instance the leakage channels satisfy the following property: there exists a state $\omega$ attaining the supremum in $\sup_\omega H_\alpha(\lkE_j)_{\lchann_j[\omega]}$, such that 
$\mathcal{P}\circ\lchann_j[\omega]$ is also a possible output state of $\lchann_j$ (i.e.~basically that the leakage channels can also produce a classical version of an output state attaining the supremum). 
We now discuss the details of how to bound the above optimization when choosing $\dom$ to encode either a dimension bound or an energy bound.

\subsection{Dimension bounds}
\label{sec:dimbnd}

\newcommand{\kmax}{{k_\mathrm{max}}}
Here, we shall suppose that we are given some constant $\dL \in \mathbb{N}$ such that every $L_j$ register has dimension at most $\dL$, i.e.~so the domain $\dom$ in~\eqref{eq:Renyiopt} is the set of $\dL$-dimensional probability distributions. We first remark that this constraint could be motivated, for instance, if we suppose that the leakage register $\lkE_j$ in each round is produced by some channel acting on a classical memory register $C_j$ with dimension upper bounded by some constant $d_C\in\mathbb{N}$, i.e.~a form of bounded-memory constraint.
In that case, one can show that without loss of generality we can 
set $d_L = d_C + 1$: intuitively, this is because we only need that many dimensions to ``encode the information'' in $C_j$ while preserving the leakage constraints; we formalize this model and present the rigorous details in Appendix~\ref{app:membnd} (there are some subtleties, e.g.~we need to start by analyzing $\Hmin^{\esp}\left(\str{S}|\str{\lk}^{A\to E} \str{\lk}^{B\to E} \En \str{P} \right)_{\rhoPE}$ directly rather than the optimization~\eqref{eq:Renyiopt}).
As an example of a possible choice for $d_C$, focusing on the case where $\str{\lkE} = \str{\lk}^{A\to E}$, we could for instance suppose that $C_j$ only stores the past $\kmax$ outputs from Alice's device for some fixed $\kmax$, in which case we have $d_C = d_A^\kmax$ where $d_A$ is the dimension of Alice's single-round output (or if we want to allow $C_j$ to store the past $\kmax$ outputs from both devices, just take $d_C = (d_A d_B)^\kmax$ instead).

We also highlight that such a dimension constraint is rather different from the one analyzed in the leakage model of~\cite{arx_JK21} --- in that work, to obtain nontrivial results, the \emph{total} dimension of all the leakage registers over the protocol must be small (more precisely, while the log-dimension can be of order $\Omega(n)$, it must be strictly less than the amount of smooth min-entropy the devices would have generated without leakage). In our model, however, we can allow the total log-dimension of $\str{\lkE}$ to be much larger than $\Hmin^{\es}(\str{S}|\En \str{P})_{\rhoPE}$, e.g.~we can set $\dL$ to be larger than the dimension of an $S_j$ register (as would be the case if we choose $d_L = d_A^\kmax + 1$ following the above bounded-memory discussion), and still obtain nontrivial results.\footnote{Still, qualitatively speaking it seems that a potential alternative approach in our setting might have been to use our leakage constraints to argue that in a classical sense, with high probability ``not too many'' of the $\lkE_j$ registers are in a nontrivial state, and use this to bound the dimension of the support of the state on $\str{\lkE}$ (given a dimension bound $\dL$), then apply the analysis in~\cite{arx_JK21}. However, it currently does not seem straightforward to formalize this into a rigorous argument when accounting for non-IID behaviour. 
Although, we highlight that in any case $\Hmax^{\esL}\left(\str{\lkE}\right)$ approximately characterizes the size of the support of the state (see e.g.~\cite{TL17}), so the approach we present here could already be viewed in some sense as one approach to formalize this intuition.} 

We now turn to the main goal of upper-bounding~\eqref{eq:Renyiopt} given some dimension bound $\dL$. 
This is in fact straightforward: intuitively, the maximum entropy should be achieved by setting $w_0$ to be as low as possible and then distributing the remaining probability uniformly over the other $\dL-1$ variables $w_k$ (assuming that $\dleak < 1-1/\dL$; otherwise we can just set $\mathbf{w}$ to be the uniform distribution and attain the trivial maximum value $H_\alpha(\lkE_j) = \log\dL$ for the {\Renyi} entropy), i.e.~set $w_0 = 1-\dleak$ and $w_k = \dleak/(\dL-1)$ otherwise.
This can be quickly confirmed by a symmetry argument: since the optimization is invariant under permutations of the variables $\{w_k | k\neq 0\}$, and the objective function is concave, the maximum value can always be attained by some solution in which all $\{w_k | k\neq 0\}$ have the same value, from which the claim easily follows. (Alternatively, one could use a Lagrange-multiplier argument; see Appendix~\ref{app:lagrange}.) Hence the optimization~\eqref{eq:Renyiopt} evaluates to
\begin{align}
(1-\dleak)^\alpha + (\dL-1)\left(\frac{\dleak}{\dL-1}\right)^\alpha
,
\end{align}
as long as $\dleak < 1-1/\dL$ (otherwise the optimal value becomes just the trivial maximum value and we cannot obtain any useful results; though in any case, for any scenario where we can expect to obtain nontrivial results we almost certainly have $\dleak \leq 1/2$ and thus $\dleak \leq 1-1/\dL$).

\begin{figure}
\centering
\includegraphics[width=0.6\textwidth]{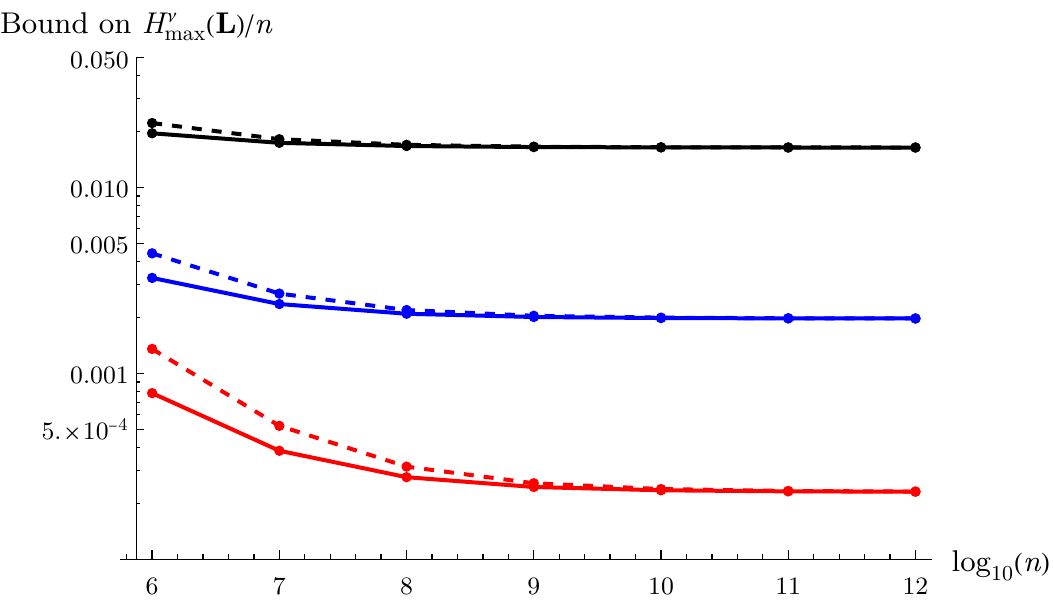}
\caption{
Upper bounds (as a function of the number of protocol rounds $n$) on $\frac{1}{n}\Hmax^{\esL}\left(\str{\lkE}\right)_{\rhoPE}$ via~\eqref{eq:Hmaxdimbnd}, numerically optimized over $\alpha\in[1/2,1)$.
The black, blue and red colours respectively denote $\dleak = 10^{-3},10^{-4},10^{-5}$, while the solid and dashed lines respectively denote the choices $\esL=\epsPE=10^{-3}$ and $\esL=\epsPE=10^{-10}$. 
We choose the dimension bound to be $\dL=2^5+1$, which can be motivated for instance by supposing that each leakage register $\lkE_j$ is produced from a classical memory register of no more than $5$ bits (see discussion in main text).
}
\label{fig:bnddim}
\end{figure}

Putting this together with~\eqref{eq:chainRenyi}--\eqref{eq:pinchbnd}, we obtain an explicit upper bound on $\Hmax^{\esL}\left(\str{\lkE}\right)_{\rhoPE}$, for $\dleak < 1-1/\dL$ and any $\alpha\in[1/2,1)$:
\begin{align}
\Hmax^{\esL}\left(\str{\lkE}\right)_{\rhoPE} \leq \frac{1}{1-\alpha} \log \left((1-\dleak)^\alpha + (\dL-1)\left(\frac{\dleak}{\dL-1}\right)^\alpha\right)n +
\frac{\log(1/\epsPE)+\smf{\esL}}{1/\alpha - 1} .
\label{eq:Hmaxdimbnd}
\end{align}
To provide a simple example calculation of the above bound, we plot it as a function of $n$ for various choices of $\dleak$ and the parameters $\esL,\epsPE$ in Fig.~\ref{fig:bnddim}, with numerical optimization of the {\Renyi} parameter $\alpha\in[1/2,1)$. We choose $\dL = 2^5+1$, which in terms of our above discussion regarding bounded memory, we can view as supposing each leakage register $\lkE_j$ is produced from a classical memory register of no more than $5$ bits (alternatively we can just suppose the leakage registers inherently have maximum dimension $\dL$ for some physical reason).
We see that it is possible to obtain nontrivial bounds on $\frac{1}{n}\Hmax^{\esL}\left(\str{\lkE}\right)_{\rhoPE}$ for these parameter choices, and also that the value becomes close to the asymptotic value (which should be $\sup_\omega H(\lkE_j)_{\lchann_j[\omega]}$) at approximately $n\sim10^9$. 

\subsection{Energy bounds}
\label{sec:Ebnd}

If we do not wish to impose a ``hard'' dimension bound on the leakage registers (or the registers they are generated from), we can instead follow a common approach for avoiding such dimension bounds, namely to impose an upper bound $\avgE$ on the expectation value of the energy with respect to some Hamiltonian --- this informally ensures that the state cannot have too much ``weight'' on high energy levels.\footnote{In the analysis for this section we will allow the leakage registers to have infinite dimension. Technically, this is not precisely consistent with our generic assumption in this work that all systems have finite (though possibly unknown) dimension, but we can just view this as a minor relaxation of that condition within this section.} At first glance, it would seem that with this we could argue that the state is close to one with support on some finite-dimensional low-energy subspace, then apply the dimension-bound analysis in the previous section. Unfortunately, this argument does not quite work out because without further structure in the Hamiltonian, for any fixed $\delta>0$ it is possible for a state on an infinite-dimensional space to be $\delta$-close (in e.g.~trace distance) to one with finite-dimensional support and yet have arbitrarily high entropy, i.e.~we do not have uniform continuity of entropy for infinite-dimensional systems.\footnote{This can be seen from e.g.~the discussion in~\cite{Win16}, but for a concrete example in our context, take any Hamiltonian with infinitely many energy levels below some finite value $E_{\star}$. Note that there always exists some sufficiently small $t>0$ such that setting $w_0 = 1-t$ and distributing the remaining weight in any fashion across the energy levels below $E_{\star}$ will still satisfy the energy constraint. 
Given there are infinitely many such energy levels, we can distribute this weight uniformly across arbitrarily many of them, which yields arbitrarily high entropy (and in fact this allows the ``random full-leakage attack'' described in Sec.~\ref{sec:leakmodel}).} Hence in this section we analyze how to tackle the optimization~\eqref{eq:Renyiopt} directly under an energy constraint, rather than attempting to relate it to some intermediate state with finite-dimensional support.

We briefly remark on some possible alternatives: another approach could have been to use the continuity bound in~\cite{Win16} based on an energy bound (and some properties of the Hamiltonian) to formalize the above argument involving an intermediate state with finite-dimensional support. However, the continuity bound in~\cite{Win16} is for the von Neumann entropy, and hence in our context we would need to either generalize the proof to {\Renyi} entropy, or bound the difference between the von Neumann entropy and {\Renyi} entropy under an energy bound (again, if the dimension is bounded, this follows from e.g.~the dimension-dependent continuity bound in~\cite{DFR20}, but with only an energy bound it is less clear how to proceed). In any case, this approach seems more indirect than what we use below, and hence is likely to yield worse bounds.
Another potential approach could be to try using the energy bound to argue that the \emph{entire} state produced by the protocol has e.g.~trace distance $\tilde{\eps}$ with respect to another state where all the $\lkE_j$ registers have bounded dimension. In that case, if we use the bounded-dimension analysis to prove that the latter produces an $\esecr$-secret key (see e.g.~\cite{arx_PR21,arx_tanthesis} for a detailed definition and discussion of $\esecr$-secrecy), we can conclude that the original state produces an $(\esecr+\tilde{\eps})$-secret key by the triangle inequality. This approach bypasses the continuity-bound obstacles described above, because the $\tilde{\eps}$-closeness is used directly to bound the trace distance to some final ideal state produced at the very end of the protocol, rather than to invoke an continuity argument in the intermediate entropic analysis. However, as we discuss in Appendix~\ref{app:dimfromE}, the scaling of the dimension bounds we could obtain from such an argument seems unlikely to be useful.

We now present our approach for handling the optimization~\eqref{eq:Renyiopt} with an energy bound.
We require the Hamiltonian $H$ to have the property that one of the
ground states is 
equal to the state $\ket{\gnd}$ we defined in our leakage model, and that the system has countable (possibly infinite) dimension.
With this, we can take the eigenbasis of the Hamiltonian as the orthonormal basis $\big\{\ket{e_k}_{\lkE_j}\big\}$ used to construct the optimization~\eqref{eq:Renyiopt}, ordering the eigenvectors such that $\ket{e_0} = \ket{\gnd}$. 
We denote the energy eigenvalues as $E_k$, and for ease of presentation we set the ground state energy $E_0$ to be zero without loss of generality.
The energy bound we shall consider is to say that for any state $\rho$ that could be produced by the leakage channel $\lchann_j$, the expectation value of the energy with respect to $H$ (i.e.~$\tr{\rho H}$) is upper bounded by some constant $\avgE\geq0$. Recalling how the $w_k$ variables in~\eqref{eq:Renyiopt} were defined in terms of the basis $\big\{\ket{e_k}_{\lkE_j}\big\}$, we have $\tr{\rho H} = \sum_k w_k E_k$. 
With this, the optimization~\eqref{eq:Renyiopt} can be written as follows (for use in our later analysis, we now write out the normalization condition as an explicit constraint):
\begin{align}
\label{eq:Ebndopt}
\begin{gathered}
\sup_{
w_k \in \mathbb{R}_{\geq 0}
} 
\sum_k w_k^\alpha\\
\begin{aligned}
\suchthat \quad & w_0 \geq 1-\dleak, \quad \sum_k w_k E_k \leq \avgE, 
\quad \sum_k w_k = 1,
\end{aligned}
\end{gathered}
\end{align}
for a countable (possibly infinite) number of variables $w_k$. (Note that even with the conditions listed above on the Hamiltonian, it is still possible for the optimization~\eqref{eq:Ebndopt} to be unbounded; e.g.~as a trivial example, if it has an infinite number of energy levels below $\avgE$. However, we shall not attempt to impose further conditions here, instead leaving it up to individual applications whether the Hamiltonian of interest yields a finite value in the optimization.)

The optimization~\eqref{eq:Ebndopt} is essentially just an entropy maximization problem subject to an energy constraint (and a ground-state constraint) --- this is similar to standard questions in thermodynamics, except that here we are considering a {\Renyi} entropy rather than Shannon/Boltzmann entropy. Hence we can apply analogous approaches to tackle the optimization; specifically, here we use a Lagrange dual analysis. We present the details in Appendix~\ref{app:lagrange}, with the main result being the following upper bound on~\eqref{eq:Ebndopt}. (In particular, we remark that the parameter $\lagE$ in this bound is a Lagrange dual variable for the energy constraint, somewhat analogous to the thermodynamic inverse-temperature parameter $\beta=1/(k_B T)$; hence our choice of notation.) In the following lemma, the optimization~\eqref{eq:Ebndopt} and the function $\dualf(\lagG, \lagE, \lagP)$ should be understood to have values in the extended reals $\mathbb{R} \cup \{\pm \infty\}$, following standard conventions in optimization theory (i.e.~for instance they take value $+\infty$ if~\eqref{eq:Ebndopt} is unbounded above or if the summation in~\eqref{eq:dualformula} diverges). 
\begin{lemma}\label{lem:Ebnd}
For any values $\lagG, \lagE, \lagP \in \mathbb{R}$ such that $\lagE \geq 0$ and $\lagP > \lagG \geq 0$, the optimization~\eqref{eq:Ebndopt} is upper bounded by
\begin{align}\label{eq:dualformula}
\dualf(\lagG, \lagE, \lagP) \defvar
(1-\alpha) \left(\left(\frac{\alpha}{\lagP - \lagG}\right)^{\frac{\alpha}{1-\alpha}}
+ \sum_{k\neq 0} \left(\frac{\alpha}{\lagE E_k + \lagP}\right)^{\frac{\alpha}{1-\alpha}}\right) 
- \lagG (1-\dleak) + \lagE \avgE + \lagP 
.
\end{align} 
Furthermore, $\dualf$ is a convex function of $(\lagG, \lagE, \lagP)$, and as long as $\dleak,\avgE>0$, it yields a tight bound in the sense that the optimal value of~\eqref{eq:Ebndopt} is equal to
\begin{align}\label{eq:dualopt}
\inf_{\lagE \geq 0, \; \lagP > \lagG \geq 0}
\dualf(\lagG, \lagE, \lagP).
\end{align} 
\end{lemma}
\noindent
This gives us a variational method to bound the optimization~\eqref{eq:Ebndopt}, by simply optimizing over the choice of $\lagG, \lagE, \lagP$; furthermore, since the optimization~\eqref{eq:dualopt} is convex, it should be well-behaved under heuristic numerical methods.
We remark that for the scenarios we studied in this work, we found that the optimal value of $\lagE$ could be very small, and hence numerical stability was improved by reparametrizing it as $\lagE = 10^{-z}$ and optimizing over $z$ instead. (With this reparametrization the optimization~\eqref{eq:dualopt} may become nonconvex; however, since $10^{-z}$ is a strictly monotone function this reparametrization preserves the property that any local minimum is a global minimum, so heuristic numerical methods should still perform well.)

\newcommand{\lvl}{l}
\newcommand{\pow}{s}
As a demonstration of our method, we now analyze the case of a harmonic oscillator Hamiltonian with $M$ independent modes\footnote{Here we suppose that $M$ is a finite fixed value, since again, if there are arbitrarily many distinguishable modes then in principle they could be used to encode enough information to allow the ``random full-leakage attack'', even under the $\dleak$ constraint. Ruling out that attack under such conditions would require more constraints on the Hamiltonian and/or states (for instance, that only a finite number $M$ of the modes are ``nontrivial'', i.e.~only those modes can have correlations with the secret data, in which case the analysis in this section can indeed be applied).} and a common ground state $\ket{\gnd}$ for all the modes, i.e.~if for each mode $m$ we use $\Delta_m$ to denote the energy level spacing and $\ket{e_{m,\lvl}}$ to denote the eigenstate for the $\lvl^\text{th}$ energy level, it has the form
\begin{align}
H &= E_0\pure{\gnd} + \sum_{m=1}^M \sum_{\lvl=1}^\infty \lvl \Delta_m \pure{e_{m,\lvl}} 
\nonumber\\
&= \sum_{m=1}^M \sum_{\lvl=1}^\infty \lvl \Delta_m \pure{e_{m,\lvl}} \quad \text{since we set } E_0=0.
\end{align}
(In terms of the notation in~\eqref{eq:Ebndopt}, we are taking the index $k$ to have values in $\{0\} \cup \left(\{1,2,\dots,M\} \times \mathbb{N}\right)$, where $k=0$ labels the ground state $\pure{\gnd}$ and otherwise $k=(m,\lvl)$ labels the $\lvl^\text{th}$ energy level of mode $m$.)

For this Hamiltonian, for each mode $m$ the infinite sum in~\eqref{eq:dualformula} takes the form
$\sum_{\lvl=1}^\infty \left(\frac{\alpha}{\lagE \lvl \Delta_m + \lagP}\right)^\pow$ where the exponent $\pow \defvar \frac{\alpha}{1-\alpha}$ is larger than $1$ (and $\alpha,\lagE,\Delta_m,\lagP\geq0$), hence the sum indeed converges as long as $\lagE>0$. 
If we consider the Hurwitz zeta function $\zeta(s,x)\defvar\sum_{\lvl=0}^\infty (\lvl+x)^{-s}$, which some computational software programs can evaluate using inbuilt methods, we could write this sum as
$\left(\frac{\alpha}{\lagE \Delta_m}\right)^\pow 
\left(\zeta\!\left(\pow,\frac{\lagP}{\lagE \Delta_m}\right)
- \left(\frac{\lagE \Delta_m}{\lagP}\right)^\pow \right)
$. However, for more flexibility (and numerical stability, as we find that the Hurwitz zeta computation can be unstable in some parameter regimes) we can instead use a simple upper bound in terms of the integral version of the sum 
(since the terms in the sum can be written as a decreasing function of $\lvl\in\mathbb{R}_{\geq 0}$): 
\begin{align}
\sum_{\lvl=1}^\infty \left(\frac{\alpha}{\lagE \lvl \Delta_m + \lagP}\right)^\pow
&\leq \left(\frac{\alpha}{\lagE \Delta_m + \lagP}\right)^\pow + \int_{\lvl=1}^\infty \left(\frac{\alpha}{\lagE \lvl \Delta_m + \lagP}\right)^\pow \mathrm{d}\lvl \nonumber\\
&= \left(\frac{\alpha}{\lagE \Delta_m + \lagP}\right)^\pow
+\left[
\frac{1}{-\pow+1} \frac{\alpha}{\lagE \Delta_m} \left(\frac{\lagE \lvl \Delta_m + \lagP}{\alpha}\right)^{-\pow+1} 
\right]_{\lvl=1}^\infty 
\nonumber\\
&= \left(\frac{\alpha}{\lagE \Delta_m + \lagP}\right)^{\frac{\alpha}{1-\alpha}} + 
\left(\frac{1-\alpha}{2\alpha-1}\right) \frac{\alpha}{\lagE \Delta_m} \left(\frac{\alpha}{\lagE \Delta_m + \lagP}\right)^{\frac{2\alpha-1}{1-\alpha}}. \label{eq:intbnd}
\end{align}
Note that conversely, the sum is lower-bounded by the above expression with the $\left(\frac{\alpha}{\lagE \Delta_m + \lagP}\right)^{\frac{\alpha}{1-\alpha}}$ term omitted, hence giving an estimate of the tightness of this bound. We find that for the examples considered below, this term is indeed very small, indicating the above bound is quite tight.

With this, we compute upper bounds on $\Hmax^{\esL}\left(\str{\lkE}\right)_{\rhoPE}$ for some parameter choices.
We first remark that the formula~\eqref{eq:intbnd} unfortunately appears to be numerically unstable if $\alpha$ is close to $1$, and hence we were unable to optimize over all $\alpha\in[1/2,1)$ as we did in the previous section. Instead we simply chose a few fixed values of $\alpha$, and computed the corresponding bounds on $\sup_\omega H_\alpha(\lkE_j)_{\lchann_j[\omega]}$; 
recalling our discussion below the bound~\eqref{eq:chainRenyi}, this means the resulting bound on $\frac{1}{n}\Hmax^{\esL}\left(\str{\lkE}\right)_{\rhoPE}$ asymptotically approaches the $\alpha$-{\Renyi} entropy for the chosen value of $\alpha$, rather than the von Neumann entropy (though as previously observed, this is not unexpected since in principle there could be Hamiltonians such that the difference between the $\alpha$-{\Renyi} entropy and von Neumann entropy in this optimization is unbounded). Still, we find that for some $\alpha$ choices it is possible to obtain reasonable results (also, we compute the corresponding values for the $(\log(1/\epsPE)+\smf{\esL})/({1/\alpha - 1})$ term to verify that it is not unreasonably large).

Specifically, we take the example of a Hamiltonian with two modes, with energy level spacings $\Delta_1 = u$ and $\Delta_2 = 2u$ for some arbitrary energy unit $u$ (as one would expect, the choice of unit does not affect the final bound, as it corresponds to just rescaling the parameter $\lagE$ in~\eqref{eq:dualformula}).
We compute results for $\alpha \in \{0.9,0.99,0.999\}$, for which the corresponding values of $1/({1/\alpha - 1})$ are $9$, $99$, and $999$ respectively (basically, when $\alpha=1-x$ for some $x\in(0,1)$, the value is $
1/x - 1
$). These values (even after multiplying by the numerator $\log(1/\epsPE)+\smf{\esL}$ of that term in~\eqref{eq:chainRenyi}) should not be too large compared to the $\Omega(n)$ smooth min-entropy term, for instance in photonic implementations which may have $n\geq10^{11}$~\cite{LLR+21}. 
With this model, we obtain the following upper bounds on $\sup_\omega H_\alpha(\lkE_j)_{\lchann_j[\omega]}$ for $\avgE = 10^{5}u$:

\def\arraystretch{1.5} 
\setlength\tabcolsep{2mm}
\centerline{
\begin{tabular}{c | c c c}
{} & $\alpha=0.9$ & $\alpha=0.99$ & $\alpha=0.999$ \\
\hline
$\dleak = 10^{-2}$ & $1.14987$ & $0.37129$ & $0.33713$ \\
$\dleak = 10^{-3}$ & $0.19596$ & $0.04566$ & $0.04053$ \\
$\dleak = 10^{-4}$ & $0.03199$ & $0.00545$ & $0.00476$
\end{tabular}
}
\def\arraystretch{1} 
\vspace{5mm}
\noindent and for $\avgE = 10^{12}u$:

\def\arraystretch{1.5} 
\setlength\tabcolsep{2mm}
\centerline{
\begin{tabular}{c | c c c}
{} & $\alpha=0.9$ & $\alpha=0.99$ & $\alpha=0.999$ \\
\hline
$\dleak = 10^{-2}$ & $5.37979$ & $0.68498$ & $0.57672$ \\
$\dleak = 10^{-3}$ & $1.00513$ & $0.07861$ & $0.06460$ \\
$\dleak = 10^{-4}$ & $0.16480$ & $0.00890$ & $0.00715$
\end{tabular}
}
\def\arraystretch{1} 
\vspace{5mm}
We see that despite the large $\avgE$ values, we can still obtain nontrivial bounds on $\sup_\omega H_\alpha(\lkE_j)_{\lchann_j[\omega]}$ for smaller values of $\dleak$ and/or values of $\alpha$ closer to $1$. We leave for future work the topic of choosing more specialized Hamiltonians tailored for specific implementations, and finding if they yield nontrivial bounds.

\section{Conclusion and further work}
\label{sec:conclusion}

In this work, we have provided techniques to compute lower bounds on the achievable key lengths in DI protocols with constrained leakage, covering both the analysis of single rounds and the required steps to obtain a finite-size security proof without an IID assumption. While we have not considered specific implementations in detail, the techniques we provide are intended to be flexible and easily built into the existing proof techniques, with the exact parameter choices being fine-tuned for individual implementations. Our results suggest that the existing DI protocol implementations should be robust against a small amount of leakage from the devices, although with our current proof techniques, we may require the leakage parameter $\dleak$ to have rather small values to obtain nontrivial results (especially for the bounded-weight leakage model). However, we highlight that for the bounded-weight model, there is still room for potential sharpening of the bounds; specifically, the continuity bound~\eqref{eq:fcont} could be significantly sharpened if one were to find a continuity bound based directly on fidelity instead of trace distance.

As a possible extension or variant of the approach presented here, we note that recently, a privacy amplification theorem was developed in~\cite{arx_Dup21} based on {\Renyi} entropy rather than smooth min-entropy. If desired, it seems possible to implement our approach in a proof based on that theorem as well. Specifically, there exists a powerful chain rule for (appropriately defined) conditional {\Renyi} entropies~\cite{Dup15}: for any $\alpha,\alpha',\alpha'' \in (1/2,1) \cup (1,\infty)$ such that $\frac{\alpha}{\alpha-1} = \frac{\alpha'}{\alpha'-1} + \frac{\alpha''}{\alpha''-1}$, one has $H_\alpha(QQ'|Q'') \geq H_{\alpha'}(Q|Q'Q'') + H_{\alpha''}(Q'|Q'')$ if $(\alpha-1)(\alpha'-1)(\alpha''-1)>0$, and the inequality is reversed if $(\alpha-1)(\alpha'-1)(\alpha''-1)<0$. By using this chain rule, we should be able to obtain an analogue of the bounds~\eqref{eq:chainHminHmax}--\eqref{eq:chainleak} here, hence lower-bounding $H_\alpha \left(\str{S}|\str{\lk}^{A\to E} \str{\lk}^{B\to E} \En \str{P} \right) $ in terms of $H_{\overline{\alpha}} (\str{S}|\En \str{P} )$ and $H_{\hat{\alpha}} \left(\str{\lk}^{A\to E}\right), H_{\hat{\alpha}}\left(\str{\lk}^{B\to E}\right)$ for some $\overline{\alpha},\hat{\alpha}$. 
Now, while in this work we have described the EAT as bounding $\Hmin^{\es}\left(\str{S}| \En \str{P} \right)$, in fact it more fundamentally provides a lower bound on the corresponding {\Renyi} entropy, so that would let us handle $H_{\overline{\alpha}} (\str{S}|\En \str{P} )$. Similarly, our approach for bounding $\Hmax^{\esL}\left(\str{\lkE}\right)$ in this work also proceeds via bounding the {\Renyi} entropy, so it already provides a bound on the $H_{\hat{\alpha}} \left(\str{\lk}^{A\to E}\right), H_{\hat{\alpha}}\left(\str{\lk}^{B\to E}\right)$ terms. Hence this approach seems plausible (in fact, this chain rule has already been applied in part of a security proof for device-dependent QKD in~\cite{arx_GLT+22}), though we leave the details and potential numerical comparison for future work.

\section*{Acknowledgements}
We thank Jean-Daniel Bancal, Peter Brown, Christopher Chubb, Omar Fawzi, Srijita Kundu, Tony Metger, Joseph Renes, Renato Renner, Nicolas Sangouard, Pavel Sekatski, and Marco Tomamichel for helpful discussions.

Financial support for this work has been provided by the Natural Sciences and Engineering Research Council of Canada (NSERC) Alliance, and Huawei Technologies Canada Co., Ltd.

Computations were performed using the MATLAB package YALMIP~\cite{yalmip} with the solver MOSEK~\cite{mosek}, as well as Mathematica.

\appendix

\section{Minor variations}
\label{app:vars}

\subsection{Event ordering in leakage channel}
\label{app:varorder}

One slightly restrictive property of the structure we have imposed on the leakage process is that (focusing on Alice; the situation for Bob is analogous) the leakage register $\lk^{A\to E}_j$ for Eve is produced \emph{before} Alice's device receives $\lk^{B\to A}_j$ from Bob and measures it to produce an output. 
This is mainly to ensure we have a well-defined joint state on $\lk^{A\to B}_j \lk^{A\to E}_j \lk^{B\to A}_j \lk^{B\to E}_j$ after applying $\lchann_j$, but has the drawback that $\lk^{A\to E}_j$ may not be fully able to encode information about Alice's output in that round. Still, this seems to be a not very significant restriction, since Alice's device can for instance produce some ``preliminary'' output $\hat{A}_j$ using only $\qA_j X_j$, then use $\lk^{A\to E}_j$ to encode some information about this preliminary output at least (although indeed this preliminary output may differ from the final output $A_j$ produced after $\lk^{B\to A}_j$ is received). Alternatively, if the device memories can retain the outputs for one round at least, the next round's registers $\lk^{A\to E}_{j+1}$ could be used to leak information about the $j^\text{th}$ round outputs, i.e.~the devices could just ``defer'' the leakage by one round. Hence when the number of rounds is large, this issue seems unlikely to be significant.
A perhaps more significant restriction in the model is that only one leakage register is sent in each direction (per round), rather than allowing for arbitrarily many iterations of leakage between the devices in both directions. 

Still, we can in fact somewhat accommodate the above possibilities while retaining the results we derived. To do so, we could instead allow the leakage channels $\lchann_j$ to have the following structure.
After the state preparation process, Alice's device performs some ``preliminary'' operations on $\qA_j X_j$ to produce a state on some registers $\qA_j \lk^{A}_j X_j$ without disturbing $X_j$. Analogously, Bob's device acts on $\qB_j Y_j$ and produces a state on registers $\qB_j \lk^{B}_j Y_j$. Then some channel is applied on the registers $\lk^{A}_j \lk^{B}_j$ to produce registers $\lk^{A\to B}_j \lk^{A\to E}_j \lk^{B\to A}_j \lk^{B\to E}_j$. Note that this last channel does not need to act locally in Alice and Bob's devices; it can act arbitrarily across the registers $\lk^{A}_j \lk^{B}_j$. 

This model is more general than the main one we focused on in this work, since it could in theory even model multiple rounds of interaction between the devices in the last step.
However, all the bounds we compute in this work apply to this model as well (under a bounded-weight or classical-probabilistic leakage constraint on the overall channel $\lchann_j$). This is because for instance in the Sec.~\ref{sec:1rnd} analysis, considering the state obtained by tracing out $\lk^{A\to B}_j \lk^{A\to E}_j \lk^{B\to A}_j \lk^{B\to E}_j$ from the output of $\lchann_j$ is equivalent to considering the state obtained by tracing out $\lk^{A}_j \lk^{B}_j$ from the output of the ``preliminary'' operations in this model. Since these operations act locally on Alice and Bob's systems, when $\lk^{A}_j \lk^{B}_j$ are removed from their outputs we can absorb the effects of these operations into the state before the measurement, in which case the remainder of the analysis holds by the same arguments. As for the Sec.~\ref{sec:fullprot} analysis, it only used the fact that the $\lk^{A\to E}_j \lk^{B\to E}_j$ registers are subject to the $\dleak$ constraint, which is the same in this model.

One potential drawback here is that interpreting the $\dleak$ constraint in this model seems less straightforward, since the registers $\lk^{A\to B}_j \lk^{A\to E}_j \lk^{B\to A}_j \lk^{B\to E}_j$ produced by such a model seem less directly related to the physical registers being sent between the devices during the actual leakage process --- they are more of a ``summary'' of the final result. On the other hand, if the physical setup justifies imposing the leakage constraints we have used on the ``preliminary registers'' $\lk^{A}_j \lk^{B}_j$  themselves, then it seems our analysis should also basically generalize to this model, by a suitable data-processing argument --- we leave the details for future work, if there is a setup which seems reasonably described by this model.

\subsection{Relations between probability bounds}
\label{app:varprob}

We first make a simple observation: 
if we have $k$ registers $Q_1 \dots Q_k$, 
then the outcome probabilities produced by performing the projective measurement $\left(\pure{\gnd}^{\otimes k}, \id-\pure{\gnd}^{\otimes k}\right)$ on these $k$ registers are the same as what we would obtain if we had measured each register individually with the projectors $\left(\pure{\gnd}, \id-\pure{\gnd}\right)$ and then coarse-grained the outcomes (in the sense that if all $k$ of the individual measurements in the latter scenario returned $\pure{\gnd}$ then we identify it with the outcome $\pure{\gnd}^{\otimes k}$ in the former scenario, and otherwise we identify it with the outcome $\id-\pure{\gnd}^{\otimes k}$). 
Note that we are only claiming that the outcome probabilities from these two processes are the same; the post-measurement states will in general be different, but we will not require them in our discussion here (our arguments only involve the outcome probabilities).

With this, to see the effect of the bounded-weight leakage constraint on e.g.~just the registers $\lk^{A\to B}_j \lk^{B\to A}_j$, we see that the above observation implies that the probability of getting the outcome $\pure{\gnd}^{\otimes 2}$ from the projective measurement $\left(\pure{\gnd}^{\otimes 2}, \id-\pure{\gnd}^{\otimes 2}\right)$ on these registers is the same as the probability of both of them giving outcome $\pure{\gnd}$ when measured individually with $\left(\pure{\gnd}, \id-\pure{\gnd}\right)$. Now note that this probability must be at least the probability of all $4$ outcomes being $\pure{\gnd}$ when individually measuring all $4$ registers $\lk^{A\to B}_j \lk^{A\to E}_j \lk^{B\to A}_j \lk^{B\to E}_j$ with $\left(\pure{\gnd}, \id-\pure{\gnd}\right)$; however, by again invoking the above observation, this is just the probability of getting the outcome $\pure{\gnd}^{\otimes 4}$ from the measurement $\left(\pure{\gnd}^{\otimes 4}, 
\id
-\pure{\gnd}^{\otimes 4}\right)$ on the registers $\lk^{A\to B}_j \lk^{A\to E}_j \lk^{B\to A}_j \lk^{B\to E}_j$.
As an alternative, one could just prove the desired result by direct calculation:
\begin{align}
\tr{\pure{\gnd}^{\otimes 2} \rho_{\lk^{A\to B}_j \lk^{B\to A}_j}} 
&= \tr{\pure{\gnd}^{\otimes 2} \otimes \left(\pure{\gnd}^{\otimes 2} + \left(\id - \pure{\gnd}^{\otimes 2}\right)\right) \rho_{\lk^{A\to B}_j \lk^{A\to E}_j \lk^{B\to A}_j \lk^{B\to E}_j}}
\nonumber \\
&\geq \tr{\pure{\gnd}^{\otimes 4} \rho_{\lk^{A\to B}_j \lk^{A\to E}_j \lk^{B\to A}_j \lk^{B\to E}_j}},
\end{align}
as claimed.

Furthermore, that observation allows us to easily analyze a minor variant of the bounded-weight leakage model, where we instead say that for each of the registers $\lk^{A\to B}_j, \lk^{A\to E}_j, \lk^{B\to A}_j, \lk^{B\to E}_j$ individually, the probability of getting outcome $\pure{\gnd}$ when measured with $\left(\pure{\gnd}, \id-\pure{\gnd}\right)$ is at least $1-\dleak'$ for some $\dleak'>0$. We shall now show the version in the main text can be straightforwardly converted into this variant and vice versa, up to a small change of the leakage parameter (for one direction of the conversion). Specifically, first note that the argument in the previous paragraph already shows
that the bounded-weight model in the main text (with leakage parameter $\dleak$) automatically implies this variant with $\dleak'=\dleak$.
As for the reverse conversion,
let us rephrase this variant as the statement that for each individual measurement the probability of getting outcome $\overline{\pure{\gnd}}$ is at most $\dleak'$, where for brevity we introduce the notation $\overline{\pure{\gnd}} \defvar \id-\pure{\gnd}$ and analogously $\overline{\pure{\gnd}^{\otimes k}} \defvar \id-\pure{\gnd}^{\otimes k}$ for any $k\in\{1,2,3,4\}$ (i.e.~this is just a compact notation for ``complementary'' outcomes). Invoking the observation in the first paragraph, the probability of getting the outcome $\overline{\pure{\gnd}^{\otimes 4}}$ from the measurement $\left(\pure{\gnd}^{\otimes 4}, \overline{\pure{\gnd}^{\otimes 4}}\right)$ on all $4$ registers is the same as the probability for individually measuring each register with $\left(\pure{\gnd}, \overline{\pure{\gnd}}\right)$ and getting at least one $\overline{\pure{\gnd}}$ outcome. 
Applying the union bound, that probability is at most 
$4\dleak'$ (note that no independence assumptions on the state across the registers are needed to apply the union bound). Hence we can conclude that this variant model implies the bounded-weight leakage model in the main text with $\dleak = 4\dleak'$. 

While the above discussion gives a simple conversion between this variant and the bounded-weight model in the main text, it is technically slightly suboptimal to just convert the former to the latter and then apply the analysis in Sec.~\ref{sec:1rnd}--\ref{sec:fullprot} --- instead, slightly sharper results for this model could be obtained by modifying the analysis in those sections appropriately, using similar arguments as those we have described above. For instance, noting that the analysis in Sec.~\ref{sec:1rnd} technically only requires analyzing the $2$ leakage registers $\lk^{A\to B}_j \lk^{B\to A}_j$, one can show that in this variant model, we can substitute $\dleak$ in the optimization~\eqref{eq:bndwtopt} with $2\dleak'$ rather than $4\dleak'$.\footnote{
However, if we are instead considering the specialized IID analysis described in Remark~\ref{remark:IIDcase}, we would still need to substitute $\dleak$ with $4\dleak'$ since in that case the $\lk^{A\to E}_j \lk^{B\to E}_j$ registers are also involved in the analysis. In that analysis, though, there is no need to separately subtract off the smooth max-entropy of $\str{\lk}^{A\to E} \str{\lk}^{B\to E}$, so no further $\dleak'$-dependent corrections are involved in that case.
}
Similarly, for Sec.~\ref{sec:fullprot}, in the optimization~\eqref{eq:Renyiopt} we can substitute $\dleak$ with $\dleak'$ rather than $4\dleak'$, since we are only considering a single register $\lk^{A\to E}_j$ or $\lk^{B\to E}_j$.

\section{Technical details regarding entropy accumulation}
\label{app:EAT}

In this appendix, we give some brief specialized comments regarding the application of the EAT in our context, assuming some background familiarity with the use of the EAT in security proofs.

Firstly, in Sec.~\ref{sec:leakmodel}, we have described Alice and Bob as announcing their inputs immediately after each round, and Eve is allowed to update her side-information in each round using those values. We remark that strictly speaking, to validly accommodate such a process, one would have to use a more recent version of the EAT~\cite{arx_MFSR22} rather than some earlier versions~\cite{DFR20,DF19} that had more restrictive conditions on the ``update structure'' of the side-information; however, all versions yield a smooth min-entropy bound of basically the same form, so it does not affect our claims in this work. (Alternatively, one could follow the approach used in the DIQKD security proofs~\cite{ARV19,TSB+22} based on the EAT versions in~\cite{DFR20,DF19}. Specifically, in those works, during the \emph{physical} protocol itself Alice and Bob do not announce their inputs until all the measurements has been performed. 
In that case the device measurements commute with the process of Eve preparing the states to send to the devices, allowing the security proof to instead be focused on a \emph{virtual} process where Eve instead prepares the entire $\En$ register before the protocol begins, without updating it based on the input values --- see e.g.~\cite{TSB+22} for further explanation. Another point worth highlighting is that there have been recent proposals for DI protocols in which the inputs are not revealed to the adversary~\cite{arx_BRC21}, though as noted in that work, for some such protocols it is not currently clear how to perform a full finite-size analysis due to some technical limitations of the EAT.)

Next, regarding the difficulty mentioned at the start of Sec.~\ref{sec:fullprot} in applying the EAT when the leakage registers are present: the main issue is that in each round, the leakage registers could potentially depend on the secret data generated in preceding rounds. 
If Eve updates her side-information in each round based on these leakage registers, this means that for the original EAT versions~\cite{DFR20,DF19}, the technique mentioned above of commuting the measurement and preparation processes no longer works. 
As for the generalized version~\cite{arx_MFSR22}, it does not seem straightforward to cleanly ``defer'' the update processes involving the leakage registers to the end in such a way that the no-signalling condition in that version is fulfilled.

We also remark that technically, the versions of the EAT proven in~\cite{DFR20,DF19,arx_MFSR22} may rely on the $\qA_j \qB_j$ systems being finite-dimensional (though the final bounds are independent of these dimensions, so we can still allow them to have unboundedly large finite dimension). However, a recent work~\cite{arx_FGR22} has extended a version of the EAT to states on general von Neumann algebras, hence it may be possible to allow the systems to actually have infinite dimension. We also note that the techniques in~\cite{arx_BFF21} for computing entropy bounds were inherently derived for infinite-dimensional systems.

Finally, a technical point regarding the derivation of the last line in~\eqref{eq:chainRenyi}. To obtain that bound from Corollary~3.5 of~\cite{DFR20}, we technically needed to use the fact that from our model of the devices, one can define some registers $R_j$, an initial state $\rho^0_{R_0}$, and a sequence of channels $\mathcal{E}_j: R_{j-1} \to R_j \lkE_j$, such that the state we consider on $\str{\lkE}$ is of the form $
(\mathcal{E}_n \circ \dots \circ \mathcal{E}_1)[\rho^0_{R_0}]$ (leaving some identity channels implicit). (Basically, the idea would be simply to encode the processes described in Sec.~\ref{sec:leakmodel} into the channels $\mathcal{E}_j$, using the registers $R_j$ to store all registers other than $\lkE_j$ after each round.)
With this we can (inductively) apply Corollary~3.5 of~\cite{DFR20}, identifying the $\lkE_j$ and $R_j$ registers in our situation with the $A_j$ and $R$ registers in that theorem statement (and setting the $B_j$ registers in that statement to be trivial registers). Strictly speaking, this would technically give us a bound where the sum in the right-hand-side of~\eqref{eq:chainRenyi} instead has terms of the form $\sup_{\omega
} H_\alpha(\lkE_j|
\lkE_1^{j-1}
)_{\mathcal{E}_j 
[\omega]}$ 
where $\omega$ is a state on $R_j 
\lkE_1^{j-1}$ (and the definition of conditional {\Renyi} entropy follows that used in~\cite{DFR20}).
However, by noting that $H_\alpha(\lkE_j|
\lkE_1^{j-1}) \leq H_\alpha(\lkE_j)$ for $\alpha\geq1/2$~\cite{Tom16}, 
and that the channels $\mathcal{E}_j$ can be defined such that they end with producing the $\lkE_j$ systems via the $\lchann_j$ channels, we can upper bound these terms with the expression in~\eqref{eq:chainRenyi}.

\section{Modified Gentle Measurement Lemma}
\label{app:GML}

\begin{lemma}
Let $\rho_{AB}$ be a state such that if a measurement with projectors $(\pure{0}, \id - \pure{0})$ (for some pure state $\ket{0}$) is performed on register $A$, the probability of getting the $\pure{0}$ outcome is at least $1-\delta$. Then we have $F(\rho_{AB}, \pure{0}_A \otimes \rho_{B}) \geq
1-\delta
$.
\end{lemma}
\begin{proof}
Let $\ket{\rho}_{ABR}$ be a purification of $\rho_{AB}$. Extend the state $\ket{0}_A$ to an orthonormal basis  
$\{\ket{j}_A\}$ 
for $A$, in which case we can write $\ket{\rho}_{ABR} = 
\sum_j \ket{j}_A  \ket{\omega^j}_{BR}
$
for some subnormalized states $\ket{\omega^j}_{BR} \defvar \left(\bra{j}_A \otimes \id_{BR}\right) \ket{\rho}_{ABR}$.
More specifically, these states have (squared) norm $\inn{\omega^j}{\omega^j} = 
\tr{\pure{j}_A {\rho}_{A}}$, which implies that in particular we have $\inn{\omega^0}{\omega^0} \geq 1-\delta$ by the condition on the state $\rho_{AB}$. 
Also observe that by tracing out $A$ from that expression for $\ket{\rho}_{ABR}$, we can write ${\rho}_{BR} = \sum_j \pure{\omega^j}_{BR}$ (this is not a spectral decomposition of ${\rho}_{BR}$ because the terms may be non-orthogonal, but this does not affect our argument). With this, by monotonicity of fidelity we have
\begin{align}
F(\rho_{AB}, \pure{0}_A \otimes \rho_{B})^2 \geq F(\rho_{ABR}, \pure{0}_A \otimes \rho_{BR})^2 
&= \bra{\rho}_{ABR} \left(\pure{0}_A \otimes \rho_{BR}\right) \ket{\rho}_{ABR} \nonumber\\
&= \bra{\omega^0}_{BR} \rho_{BR} \ket{\omega^0}_{BR} \nonumber\\
&= \sum_j \left|\inn{\omega^0}{\omega^j}\right|^2 \nonumber\\
&\geq \left|\inn{\omega^0}{\omega^0}\right|^2 \nonumber\\
&\geq (1-\delta)^2,
\end{align}
as claimed.
\end{proof}

This lemma differs slightly from the standard Gentle Measurement Lemma~\cite{Win99,Wat18} in that it does not show $\rho_{AB}$ is close to a post-measurement state (after the measurement with projectors $(\pure{0}_A \otimes \id_B, \id_{AB} - \pure{0}_A \otimes \id_B)$ is performed and the first outcome is obtained), but rather a state where the reduced state on $B$ is the \emph{same} as that before the measurement.\footnote{In our analysis we need the latter property rather than the former, because the function that sends a state to the normalized post-measurement state (conditioned on a particular outcome) is not a valid CPTP map,
which causes some problems in our argument.
One potential modification of our approach would be to instead consider the \emph{subnormalized} post-measurement state conditioned on the $\pure{\gnd}$ outcome, in which case one could write the function as a completely positive trace-nonincreasing (rather than trace-preserving) map, which does have some usable properties. However, we leave a detailed analysis of this idea for future work.
} Furthermore, the fidelity bound here is slightly worse --- in the standard Gentle Measurement Lemma (for e.g.~the version in~\cite{Wat18}), the lower bound is $\sqrt{1-\delta}$ instead. However, note that at small $\delta$ we have $\sqrt{1-\delta} \approx 1-\delta/2$, in which case the bounds are not too different (we have basically only lost about a factor of two on the $\delta$ parameter). There is also a technical restriction that in our version, we have only considered the case where the measurement outcome of interest corresponds to a rank-$1$ projector; it seems not entirely straightforward how to precisely formulate a generalization beyond this case, hence we leave it for future work.

\section{Details for relaxed optimizations}
\label{app:bndwtopt}

Recall that the leakage channel has the internal structure $\lchann = \lchann^A \otimes \lchann^B$ where $\lchann^A: \qA X \to \qA \lk^{A\to B} \lk^{A\to E} X$ and analogously for $\lchann^B$. Let us define another channel $\widetilde{\lchann}^A : \qA X \to \qA X$ where 
\begin{align}
\widetilde{\lchann}^A \defvar \pure{\gnd}_{\lk^{A\to B}} \otimes \left(\operatorname{Tr}_{\lk^{A\to B} \lk^{A\to E}} \circ \lchann^A\right),
\end{align}
i.e.~it simply discards the leakage registers from the output of $\lchann^A$ and re-initializes $\lk^{A\to B}$ in the state $\pure{\gnd}$; analogously define another channel $\widetilde{\lchann}^B$ from $\lchann^B$. 
Putting this together with~\eqref{eq:afterleak}, observe that we can write the states $\rho^{xy}_{\qA \qB XYR } \otimes \pure{\gnd}^{\otimes 2}$ from~\eqref{eq:closeness} in the form
\begin{align}
\rho^{xy}_{\qA \qB XYR } \otimes \pure{\gnd}^{\otimes 2} = (\widetilde{\lchann}^A \otimes \widetilde{\lchann}^B \otimes \idmap_R) \left[\omega_{\qA \qB R} \otimes \pure{xy}_{XY}\right].
\end{align}
Now if we were to apply the measurement channel $\mchann^A \otimes \mchann^B$ on these states, the resulting states would have the form
\begin{align}
\sigma^{xy}_{ABXYR} &=\left(\mchann^A \otimes \mchann^B \otimes \idmap_{R} \right)\left[\rho^{xy}_{\qA \qB XYR } \otimes \pure{\gnd}^{\otimes 2} \right], \nonumber \\
&= \left(\mchann^A \otimes \mchann^B \otimes \idmap_{R} \right) (\widetilde{\lchann}^A \otimes \widetilde{\lchann}^B \otimes \idmap_R) \left[\omega_{\qA \qB R} \otimes \pure{xy}_{XY}\right], \nonumber \\
&= (\widetilde{\mchann}^A \otimes \widetilde{\mchann}^B \otimes \idmap_R) \left[\omega_{\qA \qB R} \otimes \pure{xy}_{XY}\right], \label{eq:localmeas}
\end{align}
where we have defined a new channel $\widetilde{\mchann}^A \defvar \mchann^A \circ \widetilde{\lchann}^A : \qA X \to AX$, and analogously defined $\widetilde{\mchann}^B \defvar \mchann^B \circ \widetilde{\lchann}^B: \qB Y \to BY$. 
Crucially, note that these new channels are local measurement channels acting in Alice and Bob's devices respectively (and they inherit the property that they do not disturb the inputs $XY$).
In analogy to~\eqref{eq:aftermeas}, let us now also define
\begin{align}
\sigma^\mathrm{gen}_{AB XY R} = \sum_{xy} \pgen_{xy} \sigma^{xy}_{AB XY R}. \label{eq:localmeasgen}
\end{align}

Note that by monotonicity of fidelity, the bound~\eqref{eq:closeness} ensures that these new  
$\sigma$ states 
are ``close'' to the original $\rho$ states defined in~\eqref{eq:aftermeas}; more precisely, we have
\begin{align}
F\left(\rho^{xy}_{ABXYR}, \sigma^{xy}_{ABXYR} \right) \geq 1-\dleak
\quad \text{and} \quad
F\left(\rho^\mathrm{gen}_{ABXYR}, \sigma^\mathrm{gen}_{ABXYR} \right) \geq 1-\dleak
.
\end{align}
In other words, we have shown that 
each $\rho^{xy}$ state in the optimization~\eqref{eq:mainoptxy} has the property that there is a ``nearby'' state $\sigma^{xy}$ (i.e.~within fidelity $1-\dleak$) of the form given by~\eqref{eq:localmeas}, and analogously for the $\rho^\mathrm{gen}$ state in the objective function.
Therefore, we can relax the optimization~\eqref{eq:mainoptxy} to the following problem:
\begin{align}
\label{eq:middleopt}
\begin{gathered}
\inf_{\rho^{xy}, \omega, \widetilde{\mchann}^A, \widetilde{\mchann}^B} H(S|XYR)_{\rho^\mathrm{gen}}\\
\begin{aligned}
\suchthat \quad &\rho^{xy}_{AB} = \sum_{ab} \constr_{ab|xy} \pure{ab}
\text{\; and \;}
F(\rho^{xy}_{ABXYR},\sigma^{xy}_{ABXYR}) \geq 1-\dleak \quad \forall x,y .
\end{aligned} 
\end{gathered}
\end{align}
where the states $\rho^{xy}$ are now treated as optimization variables themselves (and with $\rho^\mathrm{gen} = \sum_{xy} \pgen_{xy} \rho^{xy}$), while the states $\sigma^\mathrm{gen},\sigma^{xy}$ are to be understood as functions of the state $\omega_{\qA \qB R}$ and the measurement channels $\widetilde{\mchann}^A, \widetilde{\mchann}^B$ via~\eqref{eq:localmeas}--\eqref{eq:localmeasgen}. (In this optimization, $\widetilde{\mchann}^A, \widetilde{\mchann}^B$ are allowed to be arbitrary measurement channels acting on the appropriate registers without disturbing $X,Y$. This technically means we have dropped the structure $\widetilde{\mchann}^A \defvar \mchann^A \circ \widetilde{\lchann}^A$, $\widetilde{\mchann}^B \defvar \mchann^B \circ \widetilde{\lchann}^B$ in their construction; however, since the original measurement channels $\mchann^A,\mchann^B$ were anyway arbitrary, this particular relaxation does not make a difference. The potential loss of tightness in changing from~\eqref{eq:mainoptxy} to the above optimization arises only from relaxing the exact expressions for $\rho^{xy}$ to a ``looser'' characterization in terms of nearby states.)

Finally, to further relax the optimization to our final form~\eqref{eq:bndwtopt}, we simply note that the fidelity constraints imply that the objective function is lower-bounded by $H(S|XYR)_{\sigma^\mathrm{gen}} - \fcont(\dleak)$, after which we can relax the fidelity constraints to $F(\rho^{xy}_{AB},\sigma^{xy}_{AB}) \geq 1-\dleak$ by tracing out $XYR$, and substitute in the equality constraints on $\rho^{xy}_{AB}$. This yields the optimization~\eqref{eq:bndwtopt}. In the next section, we describe how to implement the constraints in that optimization in a manner suitable for an SDP.

\subsection{Imposing fidelity constraints}
\label{app:fidSDP}

We follow the approach in e.g.~\cite{Tom16}: for any $f_{\star}\in[0,1]$ and any normalized states $\tau_Q,\sigma_Q$ on a register $Q$, if we choose some purification $\ket{\tau}_{QQ'}$ of $\tau_Q$ onto an isomorphic register $Q'$, then by Uhlmann's theorem we know that $F(\tau_Q,\sigma_Q) \geq f_{\star}$ if and only if there exists some extension $\sigma_{QQ'}$ of $\sigma_Q$ such that $\sqrt{\bra{\tau} \sigma \ket{\tau}_{QQ'}} \geq f_{\star}$. To express this in a more ``SDP-compatible'' form\footnote{To be more precise, this formulation is mainly only SDP-compatible in contexts where the state $\tau$ is a ``fixed constant'' rather than an optimization variable, and we can write a specific purification $\ket{\tau}$ of it --- if the state $\tau$ and/or its purification $\ket{\tau}$ were intended to be an optimization variable itself, then this approach would not be SDP-compatible since the quantity $\bra{\tau} \widetilde{\sigma} \ket{\tau}_{QQ'} = \tr{\pure{\tau}_{QQ'} \widetilde{\sigma}_{QQ'}}$ is not jointly affine with respect to $\widetilde{\sigma},\tau$. However, the alternative approach we present next (namely, the formulation in~\eqref{eq:fidasSDP}) would be usable in that scenario.} (and to avoid ambiguity from overloading the $\sigma$ symbol), we can rephrase this equivalently as the statement that there exists an operator $\widetilde{\sigma}_{QQ'} \geq 0
$ such that $\bra{\tau} \widetilde{\sigma} \ket{\tau}_{QQ'} \geq f_{\star}^2$ and $\tr[Q']{\widetilde{\sigma}_{QQ'}} = \sigma_Q$ (note that the latter constraint automatically imposes normalization on $\widetilde{\sigma}_{QQ'}$). 

With this in mind, we can rewrite the optimization~\eqref{eq:bndwtopt} as follows: for each $xy$, let us define the state $\tau^{xy}_{AB} \defvar \sum_{ab} \constr_{ab|xy} \pure{ab}_{AB}$, so the constraints simply become $F\!\left(\tau^{xy}_{AB}, \sigma^{xy}_{AB}\right) \geq 1-\dleak$.
Picking any particular purifications of the $\tau^{xy}$ states, for instance $\ket{\tau^{xy}}_{ABA'B'} \defvar \sum_{ab} \sqrt{\constr_{ab|xy}} \ket{abab}_{ABA'B'}$, from the above argument we see that the optimization can be equivalently written as
\begin{align}
\label{eq:optSDPfid}
\begin{gathered}
\inf_{\omega, \widetilde{\mchann}^A, \widetilde{\mchann}^B, \widetilde{\sigma}^{xy}} H(S|XYR)_{\sigma^\mathrm{gen}} - \fcont(\dleak)\\
\suchthat \quad \bra{\tau^{xy}}\widetilde{\sigma}^{xy}\ket{\tau^{xy}}_{ABA'B'} \geq (1-\dleak)^2 \text{\; and \;} \tr[A'B']{\widetilde{\sigma}^{xy}_{ABA'B'}} = \sigma^{xy}_{AB} \quad \forall x,y,
\end{gathered}
\end{align}
where $\widetilde{\sigma}^{xy}$ are states on $ABA'B'$,
and the states $\sigma^\mathrm{gen},\sigma^{xy}$ are again understood as functions of $\omega_{\qA \qB R}, \widetilde{\mchann}^A, \widetilde{\mchann}^B$ via~\eqref{eq:localmeas}--\eqref{eq:localmeasgen}.
With this formulation, the constraints can indeed be imposed in an SDP.

As an alternative approach, one could use a result derived in~\cite{Kil12,Wat12,Wat18} --- 
for two normalized quantum states $\tau,\sigma$ of dimension $d$, their fidelity $F(\tau,\sigma)=\norm{\sqrt{\tau}\sqrt{\sigma}}_1$ can be expressed as the following SDP (in which the maximum is indeed attained):
\begin{align}
\begin{gathered}
\max_{\Xi \in \mathbb{C}^{d\times d}} \frac{1}{2} \tr{\Xi+\Xi^\dagger} \\
\begin{aligned}
\suchthat \quad & 
\begin{pmatrix}
\tau & \Xi\\
\Xi^\dagger & \sigma\\
\end{pmatrix}
\geq 0.
\end{aligned}
\end{gathered}
\end{align}
(By symmetry properties of the above optimization, $\Xi$ can be restricted to be hermitian if desired, given that $\tau,\sigma$ are hermitian.) In particular, this implies that 
\begin{align}
\label{eq:fidasSDP}
F(\tau,\sigma) \geq f_{\star} \quad \iff \quad 
\exists\; \Xi \in \mathbb{C}^{d\times d}
\suchthat\; \frac{1}{2} \tr{\Xi+\Xi^\dagger}\geq f_{\star} \text{ and }
\begin{pmatrix}
\tau & \Xi\\
\Xi^\dagger & \sigma
\end{pmatrix}
\geq 0,
\end{align}
which also leads to an SDP formulation of the fidelity constraint.
However, we found that this approach seemed less numerically stable in some cases, and hence we mostly used the formulation~\eqref{eq:optSDPfid}. (In more general circumstances though, the formulation in~\eqref{eq:fidasSDP} has the advantage that the states $\tau,\sigma$ can \emph{both} be treated as optimization variables, unlike the preceding approach based on Uhlmann's theorem.)

\section{Dimension bound from memory bound}
\label{app:membnd}

We first formalize the bounded-memory constraint within our leakage model. For each round $j$, let $C_j$ denote a classical memory register that was stored within the memory registers $\mA_{j-1} \mB_{j-1}$ of the preceding round and passed forward to the registers $\qA_j \qB_j$ during the ``update'' channel; let $\widetilde{Q}^A_j \widetilde{Q}^B_j$ denote all registers in $\qA_j \qB_j$ other than $C_j$ ($\widetilde{Q}^A_j \widetilde{Q}^B_j$ can still contain another copy of $C_j$). With this, we are requiring that $\lchann_j$ has the form $\overline{\lchann}^{A \leftrightarrow B}_j \otimes \overline{\lchann}^{A\to E}_j \otimes \overline{\lchann}^{B\to E}_j$ for some channels $\overline{\lchann}^{A \leftrightarrow B}_j: \widetilde{Q}^A_j \widetilde{Q}^B_j X_j Y_j \to \widetilde{Q}^A_j \widetilde{Q}^B_j \lk^{A\to B}_j \lk^{B\to A}_j X_j Y_j$ and $\overline{\lchann}^{A\to E}_j: C_j \to \lk^{A\to E}_j$ and $\overline{\lchann}^{B\to E}_j: C_j \to \lk^{B\to E}_j$ (in a minor abuse of notation we allow both the last two channels to have input register $C_j$; this does not cause any issues because $C_j$ is a classical register and can be copied without disturbance to input into both channels).

In the remainder of this analysis, we again sometimes omit the $j$ subscripts for brevity; also, analogous to our $\lkE$ notation in Sec.~\ref{sec:fullprot}, we let $\overline{\lchann}$ denote either $\overline{\lchann}^{A\to E}_j$ or $\overline{\lchann}^{B\to E}_j$ since the analysis is the same for both.
For each classical basis state $\pure{c}_C$ of the classical register $C$, let us write $\rho^{(c)}_{\lkE} \defvar \overline{\lchann}\left[\pure{c}_C\right]$.
We first discuss the bounded-weight leakage constraint, which implies that every $\rho^{(c)}$ satisfies $F(\rho^{(c)},\pure{\gnd}) = \sqrt{\bra{\gnd}\rho^{(c)}\ket{\gnd}} \geq \sqrt{1-\dleak}$. 
Hence letting $\widetilde{\lkE}$ be a purifying register for $\lkE$, for each $c$ Uhlmann's theorem gives a purification $\ket{\rho^{(c)}}_{\lkE \widetilde{\lkE}}$ of $\rho^{(c)}_{\lkE}$ such that $F\left(\pure{\rho^{(c)}}_{\lkE \widetilde{\lkE}}, \pure{\gnd}_{\lkE} \otimes \pure{\gnd}_{\widetilde{\lkE}} \right) \geq \sqrt{1-\dleak}$ (note that here we use the same purification $\ket{\gnd}_{\lkE} \ket{\gnd}_{\widetilde{\lkE}}$ of $\ket{\gnd}_{\lkE}$ for every $c$). Then we have $\overline{\lchann} = \operatorname{Tr}_{\widetilde{\lkE}} \circ \overline{\lchann}'$ where $\overline{\lchann}'$ is a classical-to-quantum channel defined by $\overline{\lchann}'\left[\pure{c}_C\right] = \pure{\rho^{(c)}}_{\lkE \widetilde{\lkE}}$, and this channel $\overline{\lchann}'$ still satisfies the bounded-weight leakage constraint with the same $\dleak$ value (with the new ``reference state'' being $\ket{\gnd}_{\lkE} \ket{\gnd}_{\widetilde{\lkE}}$). Also, the states $\ket{\rho^{(c)}}_{\lkE \widetilde{\lkE}}$ and $\ket{\gnd}_{\lkE} \ket{\gnd}_{\widetilde{\lkE}}$ span a subspace of dimension at most $d_C+1$, and hence we can restrict the output space of $\overline{\lchann}'$ to be this subspace.

With the above properties, we see that if we were to consider a virtual process where the channels $\overline{\lchann}_j$ in the actual protocol are replaced with these new ``virtual'' channels $\overline{\lchann}'_j$, we would end up with the same reduced state on $\str{S} \str{\lk}^{A\to E} \str{\lk}^{B\to E} \En \str{P}$, so the value of $\Hmin^{\esp}\left(\str{S}|\str{\lk}^{A\to E} \str{\lk}^{B\to E} \En \str{P} \right)_{\rhoPE}$ in the actual protocol state is the same as for the state produced by the virtual process.
Furthermore, for the latter state we can write  $\Hmin^{\esp}\left(\str{S}|\str{\lk}^{A\to E} \str{\lk}^{B\to E} \En \str{P} \right)_{\rhoPE} 
\geq \Hmin^{\esp}\left(\str{S}|\str{\lk}^{A\to E} \str{\lk}^{B\to E} \widetilde{\str{\lk}}^{A\to E} \widetilde{\str{\lk}}^{B\to E} \En \str{P} \right)_{\rhoPE}$. We can then simply lower-bound the latter by viewing $\overline{\lchann}'_j$ as our leakage channels and noting that they still satisfy the original leakage constraints, and their output dimension can be restricted to $d_C+1$, so the analysis in Sec.~\ref{sec:dimbnd} can be applied.
(Note that we cannot use an argument of this form to upper bound $\Hmax^{\esL}\left(\str{\lkE}\right)_{\rhoPE}$ instead, since for quantum systems it is not always upper bounded by $\Hmax^{\esL}\left(\str{\lkE}\overline{\str{\lkE}}\right)_{\rhoPE}$ --- for our argument to work, we had to directly consider the conditional smooth min-entropy.)

As for the classical-probabilistic leakage constraint, there is a slight technicality: since it implies the bounded-weight constraint (with the same $\dleak$ value), we could apply the above argument to again show that it suffices to consider $\overline{\lchann}'$ to have an output space of dimension $d_C+1$, but with the subtlety that this channel would only be subject to the bounded-weight constraint (with parameter $\dleak$) rather than the original classical-probabilistic constraint. However, recall that the analysis in Sec.~\ref{sec:fullprot} yields the same results for either of those leakage constraints, and hence this relaxation of the type of constraint on the channel does not change the final results.

\section{Dimension bounds from energy bounds}
\label{app:dimfromE}

The idea here is to introduce some finite ``cutoff'' energy value $\cutE$, and argue that if we remove all weight on the leakage registers with energy above this cutoff, the resulting state is still close to the original state.
For this approach, we assume that the system Hamiltonians are noninteracting, so the total energy of the systems is given by just summing the energies of the individual systems.
Furthermore, we just focus on discussing either one of the leakage systems $\str{\lk}^{A\to E}$ or $\str{\lk}^{B\to E}$, since to analyze both of them we could repeat this argument for each of them and then apply the triangle inequality.

Let $\pvm_{<}$ denote the projector onto the subspace spanned by energy eigenstates $\ket{e_k}_{\lkE}$ with energy (strictly) less than $\cutE$. Consider 
measuring all the $\lkE_j$ systems individually 
using the measurement with projectors $(\pvm_{<}, \id - \pvm_{<})$, which is equivalent to measuring all the systems using the measurement with projectors $(\pvm_{<}^{\otimes n}, (\id - \pvm_{<}) \otimes \pvm_{<}^{\otimes n-1} , \pvm_{<} \otimes (\id - \pvm_{<}) \otimes \pvm_{<}^{\otimes n-2}, \dots, (\id - \pvm_{<})^{\otimes n})$.
If we could certify that the probability of getting the outcome 
$\pvm_{<}^{\otimes n}$ 
is at least some value $1-\tilde{\eps}^2$, then by the Gentle Measurement Lemma we would know that the original state is $\tilde{\eps}$-close (in purified or trace distance) to the post-measurement state conditioned on that outcome. Observe also that this post-measurement state is supported on the subspace spanned by considering only the eigenstates $\ket{e_k}_{\lkE}$ with energy less than $\cutE$, hence we could try to prove its security by applying the dimension-bound-based analysis (assuming that there are only finitely many such $\ket{e_k}_{\lkE}$).\footnote{For this sketch we gloss over the technicality mentioned in Appendix~\ref{app:GML} that sending some state to its corresponding post-measurement state does not form a valid CPTP map; it should probably be possible to address this issue with another suitably modified Gentle Measurement Lemma and/or allowing trace-nonincreasing maps.}

Thus our task is reduced to finding the probability bound $1-\tilde{\eps}^2$ as a function of $\cutE$ (and the expected-energy bound $\avgE$).
To do so, first observe that not getting the outcome 
$\pvm_{<}^{\otimes n}$ 
implies that at least one of the systems had an outcome corresponding to energy at least $\cutE$, which implies that the outcome value for total energy is at least $\cutE$ as well, recalling we chose our ground state energy such that all energies are non-negative. (Here we implicitly used the fact that the measurements $(\pvm_{<}, \id - \pvm_{<})$ produce the same outcome probabilities as performing an energy measurement and then coarse-graining the outcome depending on whether the value is below $\cutE$.) However, we know the total expected-energy of the systems is bounded by $n \avgE$, and hence by Markov's inequality, the probability of such an outcome cannot be more than $n \avgE/\cutE$. Therefore we could take $\tilde{\eps} = \sqrt{n \avgE/\cutE}$.

Unfortunately, this bound is rather trivial as it is increasing in $n$. (It might be possible to choose $\cutE \propto n$ and still obtain nontrivial results for specific Hamiltonians, but this does not appear very promising in general). The core difficulty here seems to be that while the energy bound implies that the probability of each individual register having energy above $\cutE$ is at most $\avgE/\cutE$ (by Markov's inequality), in order to apply the dimension-bound argument we need to ensure that \emph{all} the registers are subject to the cutoff. (Applying the Gentle Measurement Lemma to the $\lkE_j$ registers individually also seems unlikely to help, because this would only give an exponentially decreasing fidelity bound between the original and post-measurement states on the full $\str{\lkE}$ system.) We leave for future work the question of whether approaching the analysis in a different way could overcome the scaling issue in this approach.

\section{Lagrange dual analysis}
\label{app:lagrange}

Here we focus on proving Lemma~\ref{lem:Ebnd} for the energy-bound case (Sec.~\ref{sec:Ebnd}); the results for the dimension-bound case (Sec.~\ref{sec:dimbnd}) can be easily obtained by an analogous argument (simply omit the energy constraint and impose the fact that the number of $w_k$ variables is $\dL$). 

\newcommand{\wmin}{w_\mathrm{min}}
In this section let us use the notation $\wmin \defvar 1-\dleak$, so the first constraint in~\eqref{eq:Ebndopt} can be written more compactly as $w_0 \geq \wmin$. As mentioned above Lemma~\ref{lem:Ebnd}, we work in the extended reals $\mathbb{R} \cup \{\pm \infty\}$.
Since~\eqref{eq:Ebndopt} is a constrained optimization problem, it is easily shown (see e.g.~\cite{BV04v8}) that for any choice of 
$\lagG, \lagE \in \mathbb{R}_{\geq0}$ and $\lagP \in \mathbb{R}$,
we can upper-bound the optimal value with the \term{Lagrange dual function} of the optimization:
\begin{align}\label{eq:dualfunc}
\dualf(\lagG, \lagE, \lagP) \defvar \sup_{w_k \in \mathbb{R}_{\geq 0}} 
L(\mathbf{w}, \lagG, \lagE, \lagP)
,
\end{align} 
where $L$ is the \term{Lagrangian} (choosing a sign convention for the equality constraint that keeps the final expressions slightly cleaner):
\begin{align}
L(\mathbf{w}, \lagG, \lagE, \lagP) \defvar 
\sum_k w_k^\alpha + \lagG (w_0 - \wmin) + \lagE \left(\avgE - \sum_k w_k E_k \right) + \lagP \left(1-\sum_k w_k \right)
.
\end{align} 
We shall now show that for $\lagP > \lagG$, the above expression for $\dualf$ reduces to the one presented in Lemma~\ref{lem:Ebnd}. (For $\lagG \geq \lagP$, the above expression evaluates to $\dualf(\lagG, \lagE, \lagP) = +\infty$ and hence only yields a trivial bound --- to see this, note that $L(\mathbf{w}, \lagG, \lagE, \lagP)$ is of the form $w_0^\alpha + \lagG w_0 - \lagP w_0 + Z$ where $Z$ denotes some terms independent of $w_0$. Therefore, when $\lagG \geq \lagP$ the expression is unbounded as we take $w_0 \to +\infty$, and hence $\dualf(\lagG, \lagE, \lagP) = +\infty$. We hence do not consider this regime in the rest of our analysis; if desired, from the perspective of optimization theory we can view the function $\dualf(\lagG, \lagE, \lagP)$ in Lemma~\ref{lem:Ebnd} as being implicitly understood to have value $+\infty$ outside of the domain specified in the lemma statement.)

To simplify the expression for $\dualf$, we note that by concavity of the original optimization~\eqref{eq:Ebndopt} (or by direct inspection) $L(\mathbf{w}, \lagG, \lagE, \lagP)$ is a concave function of $\mathbf{w}$, which implies that if the domain contains a stationary point, then that point attains the maximum value over the domain. To find whether such a stationary point exists, we compute the partial derivatives $\frac{\partial L}{\partial w_k}$ over the interior of the domain, i.e.~for $w_k > 0$:
\begin{align}\label{eq:partials}
\begin{aligned}
&\frac{\partial L}{\partial w_0} = \alpha w_0^{\alpha-1} + \lagG - \lagP, \\
&\frac{\partial L}{\partial w_k} = \alpha w_k^{\alpha-1} - \lagE E_k - \lagP \quad \text{ for } k\neq 0,
\end{aligned}
\end{align}
where in the first line we have used the fact that we set $E_0=0$.
We hence see that given $\lagP > \lagG$, the system of equations $\frac{\partial L}{\partial w_k} = 0$ indeed has a solution with strictly positive $w_k$ values, namely:
\begin{align}
\begin{aligned}
&w_0 = \left(\frac{\lagP - \lagG}{\alpha}\right)^{\frac{1}{\alpha-1}} 
= \left(\frac{\alpha}{\lagP - \lagG}\right)^{\frac{1}{1-\alpha}}, 
\\
&w_k = \left(\frac{\lagE E_k + \lagP}{\alpha}\right)^{\frac{1}{\alpha-1}} 
= \left(\frac{\alpha}{\lagE E_k + \lagP}\right)^{\frac{1}{1-\alpha}}
\quad \text{ for } k\neq 0.
\end{aligned}
\end{align}

Substituting this solution into~\eqref{eq:dualfunc}, 
we conclude that for $\lagG, \lagE, \lagP$ satisfying the lemma conditions, we have
\begin{align}
\dualf(\lagG, \lagE, \lagP) &=
\left(\frac{\alpha}{\lagP - \lagG}\right)^{\frac{\alpha}{1-\alpha}}
- (\lagP-\lagG)\left(\frac{\alpha}{\lagP - \lagG}\right)^{\frac{1}{1-\alpha}} \nonumber\\
&\quad + \sum_{k\neq 0} \left( \left(\frac{\alpha}{\lagE E_k + \lagP}\right)^{\frac{\alpha}{1-\alpha}} - (\lagE E_k + \lagP) \left(\frac{\alpha}{\lagE E_k + \lagP}\right)^{\frac{1}{1-\alpha}} \right) \nonumber\\
&\quad - \lagG \wmin + \lagE \avgE + \lagP 
,
\end{align} 
which simplifies to~\eqref{eq:dualformula} after observing that 
\begin{align}
(\lagP-\lagG)\left(\frac{\alpha}{\lagP-\lagG}\right)^{\frac{1}{1-\alpha}} 
= \alpha^{\frac{1}{1-\alpha}} \left(\frac{1}{\lagP-\lagG}\right)^{\frac{\alpha}{1-\alpha}} 
= \alpha \left(\frac{\alpha}{\lagP-\lagG}\right)^{\frac{\alpha}{1-\alpha}},
\end{align}
and similarly for the terms in the summation.

As for the remaining claims in the lemma, the Lagrange dual function $\dualf$ for a constrained maximization problem is always a convex function of the dual variables~\cite{BV04v8}, since it is the supremum of a family of affine functions (of the dual variables). As for showing that the optimization~\eqref{eq:dualopt} 
has the same value as the original optimization~\eqref{eq:Ebndopt} when $\dleak>0$ (i.e.~$\wmin<1$) and $\avgE>0$, 
this means we have to show that strong duality holds for such parameter values (viewing the former as the dual optimization and the latter as the primal optimization).
In principle, we could do this by noting that the primal optimization~\eqref{eq:Ebndopt} has a strictly feasible point for those parameter values (except for some edge cases where the optimization is already unbounded)\footnote{Explicitly: first consider a simple case where $E_\mathrm{gap} \defvar \inf \{E_k | E_k > 0\}$ is strictly positive (i.e.~we have a gapped Hamiltonian) and the ground state is nondegenerate. By the countable-dimension assumption, we can use $\mathbb{N}$ to label all the non-ground-state energy levels. 
Then for any $t>0$, if we set $w_k=t/(2^k E_k) > 0$ for $k\neq0$, we have $\sum_{k=1}^\infty w_k E_k = \sum_{k=1}^\infty t/2^k = t$, and choosing $w_0$ to satisfy normalization we have $w_0 = 1-\sum_{k=1}^\infty w_k \geq 
1- t/E_\mathrm{gap}$. Hence by choosing $t$ sufficiently small we can satisfy both the $\wmin$ and $\avgE$ constraints with strict inequality (given $\wmin<1$ and $\avgE>0$), with all $w_k$ values being strictly positive, yielding a strictly feasible point.
To cover the edge cases, first note that for any Hamiltonian with $E_\mathrm{gap}=0$ there must be infinitely many energy levels arbitrarily close to zero and hence the optimization~\eqref{eq:Ebndopt} is unbounded in the first place; 
finally, if $E_\mathrm{gap}>0$ and the ground state is degenerate, then either it is infinitely degenerate (in which case~\eqref{eq:Ebndopt} is again unbounded) or it is finitely degenerate and we can just slightly modify the above construction to obtain a strictly feasible point (simply remove some weight from a sufficiently high energy level 
and redistribute it over the ground states).
}, and then invoking an appropriate generalization of Slater's condition to infinite-dimensional domains. However, to offer an alternative approach that avoids technicalities in handling infinite-dimensional vector spaces, we present below a proof for our case that bypasses this aspect, by extracting only the necessary intermediate steps from the standard proofs of Slater's condition (see e.g.~\cite{BV04v8,SB14}).

Specifically, consider the function $F(\wmin,\avgE,p)$ 
defined as the optimal value of the following concave optimization (which is just a slight generalization of the optimization~\eqref{eq:Ebndopt}, by allowing the weights $w_k$ to sum to some $p\in\mathbb{R}$ instead of $1$):
\begin{align}\label{eq:pertopt}
\begin{gathered}
\sup_{
w_k \in \mathbb{R}_{\geq 0}
} 
\sum_k w_k^\alpha\\
\begin{aligned}
\suchthat \quad & w_0 \geq \wmin, \quad \sum_k w_k E_k \leq \avgE, 
\quad \sum_k w_k = p,
\end{aligned}
\end{gathered}
\end{align}
taking the optimization to have value $-\infty$ if it is infeasible. 
The domain $\mathrm{dom}(F)$ of this function is defined~\cite{BV04v8} to be the set of values $(\wmin,\avgE,p)\in\mathbb{R}^3$ such that $F(\wmin,\avgE,p) \neq -\infty$. 
Now to show that strong duality holds for some particular choice of $(\wmin,\avgE,p)$ in this optimization, it suffices to show that this choice of $(\wmin,\avgE,p)$ lies in the interior of $\mathrm{dom}(F)$ 
(see the proofs of Slater's condition in e.g.~\cite{BV04v8,SB14}; the geometric idea is that this ensures the existence of a nonvertical supporting hyperplane of the subgraph of $F$ at that point, from which an optimal dual solution can be obtained).
In particular, we are focusing on the situation where $\wmin<1$, $\avgE>0$ and $p=1$, in which case it is straightforward to find some sufficiently small $t>0$ such that for all $(\wmin',\avgE',p')$ within distance $t$ (in some norm) of $(\wmin,\avgE,1)$, we have 
$\wmin' 
\leq p'$ 
and $\avgE',p'>0$.\footnote{
For example, we can set $t=\min\{(1-\wmin)/2,\avgE/2,1/2\}>0$ and use the $\infty$-norm, i.e.~$\max\{|\wmin'-\wmin|,|\avgE'-\avgE|,|p-1|\}$.
} The optimization is feasible for all such $(\wmin',\avgE',p')$, since there is a simple feasible point given by $w_0=p'$ and $w_k=0$ for all $k\neq0$.
Hence $F$ is finite in that neighbourhood, i.e.~$(\wmin,\avgE,1)$ is an interior point of $\mathrm{dom}(F)$, as required.

We remark that if the optimization~\eqref{eq:Ebndopt} is such that the optimal solution saturates both the inequality constraints (this will be the case in most situations; specifically, as long as the $\avgE$ value is not so low that it enforces $w_0$ to be strictly larger than $\wmin$, and the Hamiltonian is such that the maximum entropy is an increasing function of the expected 
energy), 
then we can replace them with equality constraints. In that case the theory of Lagrange multipliers (along with strict concavity of the objective function and linearity of the constraints) implies that solving the system of equations $\frac{\partial L}{\partial w_k} , \frac{\partial L}{\partial \lagG} , \frac{\partial L}{\partial \lagE} , \frac{\partial L}{\partial \lagP} = 0 $ for $\mathbf{w}, \lagG, \lagE, \lagP$ yields the optimal $\mathbf{w}$ in the primal optimization~\eqref{eq:Ebndopt} and the optimal $\lagG, \lagE, \lagP$ in the dual optimization~\eqref{eq:dualopt}.
In particular, since the equations $\frac{\partial L}{\partial \lagG} , \frac{\partial L}{\partial \lagE} , \frac{\partial L}{\partial \lagP} = 0$ just reproduce the optimization constraints, this means 
that in principle the choice of $\lagG, \lagE, \lagP$ that yields the best upper bound could be obtained by substituting the equations~\eqref{eq:partials} into the optimization constraints and solving for $\lagG, \lagE, \lagP$.
However, we currently do not have any explicit Hamiltonian in which we can solve the resulting expression (the various resulting summations can be expressed in terms of the Hurwitz zeta function or its derivatives via similar arguments as in the main text, 
but this is difficult to handle further in closed form).

We close this section with some side-remarks about the ``thermodynamic version'' of this argument, i.e.~if we had instead taken the objective function to be the Shannon/Boltzmann entropy. 
(As mentioned in the main text, a solution to this version does not currently seem usable for our context as we do not have a good method to relate it to $H_\alpha(\lkE_j)$ without a dimension bound, and in any case our above approach should give better results as it directly analyzes $H_\alpha(\lkE_j)$. Still, we mention it in case it highlights some useful properties.)
Note that we still keep the constraint on $w_0$, i.e.~we are maximizing the entropy of a system subject to a ground-state probability constraint as well as an energy constraint.
In that case, solving the system of equations  $\frac{\partial L}{\partial w_k} = 0$ yields solutions for $\{w_k | k\neq 0\}$ that are exponentially decreasing with respect to $E_k$.
If we furthermore suppose that the inequality constraints are saturated as mentioned above, this means the optimal solution is essentially a Gibbs state except with a larger value of $w_0$ due to the ground-state constraint, as one might intuitively expect.
Furthermore, for some simple Hamiltonians (such as harmonic oscillators, as studied in the main text) one can explicitly solve for the optimal Lagrange-multiplier values, with the optimal value of $\lagE$ yielding the inverse-temperature parameter $1/(k_B T)$, and the optimal value of $\lagP$ being related to the partition function.

\printbibliography

\end{document}